\documentclass[11pt,a4paper,reqno]{amsart}%
\usepackage{amsthm,amsmath,amsfonts,amssymb,amsxtra,appendix,bookmark,dsfont,bm,mathrsfs}
\usepackage[margin=30mm]{geometry}

\theoremstyle{plain}
\newtheorem{theorem}{Theorem}
\newtheorem{lemma}[theorem]{Lemma}
\newtheorem{corollary}[theorem]{Corollary}

\theoremstyle{definition}

\theoremstyle{remark}
\newtheorem{remark}[theorem]{Remark}

\DeclareMathOperator{\Tr}{Tr}

\def\leqslant{\le}
\def\bq{\begin{eqnarray}}
\def\eq{\end{eqnarray}}
\def\bqq{\begin{align*}}
\def\eqq{\end{align*}}

\def\nn{\nonumber}

\def\eps{\varepsilon}
\def\wto{\rightharpoonup}
\newcommand{\norm}[1]{\left\lVert #1 \right\rVert}

\renewcommand{\epsilon}{\varepsilon}
\newcommand\1{{\ensuremath {\mathds 1} }}

\def\cF {\mathcal{F}}

\def\cN{\mathcal{N}}

\def\R {\mathbb{R}}

\def\cG {\mathcal{G}}

\def\cE {\mathcal{E}}

\def\H{\gH}

\def\R {\mathbb{R}}

\def\d{\,{\rm d}}
\newcommand{\w}{\widetilde}
\newcommand{\gH}{\mathfrak{H}}

\newcommand{\bH}{\mathbb{H}}
\newcommand{\dGamma}{{\ensuremath{\rm d}\Gamma}}

\allowdisplaybreaks

\title[Focusing dynamics for low dimensional bosons]{Norm approximation for many-body quantum dynamics: focusing case in low dimensions}

\author[P.T. Nam]{Phan Th\`anh Nam}
\address{Department of Mathematics, LMU Munich, Theresienstrasse 39, 80333 Munich, Germany} 
\email{nam@math.lmu.de}
\author[M. Napi\'orkowski]{Marcin Napi\'orkowski}
\address{Department of Mathematical Methods in Physics, Faculty of Physics, University of Warsaw,  Pasteura 5, 02-093 Warszawa, Poland}
\email{marcin.napiorkowski@fuw.edu.pl}

\begin{document}
\date{\today}

\begin{abstract} We study the norm approximation to the Schr\"odinger dynamics of $N$ bosons in $\R^d$ ($d=1,2$) with an interaction potential of the form $N^{d\beta-1}w(N^{\beta}(x-y))$. Here we are interested in the focusing case $w\le 0$. Assuming that there is complete Bose-Einstein condensation in the initial state, we show that in the large $N$ limit, the evolution of the condensate is effectively described by a nonlinear Schr\"odinger equation and the evolution of the fluctuations around the condensate is governed by a quadratic Hamiltonian, resulting from Bogoliubov approximation. Our result holds true for all $\beta>0$ when $d=1$ and for all $0<\beta<1$ when $d=2$. 
\end{abstract}

\maketitle


\section{Introduction}

Bose-Einstein condensation is a remarkable phenomenon of dilute Bose gases at very low temperatures, when many particles occupy a common single quantum state. This was predicted in 1924 by Bose and Einstein \cite{Bose-24,Einstein-25} and produced experimentally in 1995 by Cornell, Wieman and  Ketterle \cite{CorWie-95,Ketterle-95}. Since then, there have remained fundamental questions in the rigorous understanding of the condensation and fluctuations around the condensate. The latter is essential for the emergence of superfluidity and other interesting quantum effects. 

From first principles of quantum mechanics, a system of $N$ identical (spinless) bosons in $\R^d$ is described by a wave function in the bosonic Hilbert space 
$$\H^N=\bigotimes_{\text{sym}}^N L^2(\R^d).$$
The evolution of the system is governed by the Schr\"odinger equation
\bq \label{eq:schrodingerdynamics}
i\partial_t \Psi_N(t) = H_N \Psi_N(t) 
\eq
where $H_N$ is the Hamiltonian of the system. We will focus on the typical Hamiltonian of the form\begin{equation} \label{eq:HN}
H_N= \sum\limits_{j = 1}^N -\Delta_{x_j} + \frac{1}{N-1} \sum\limits_{1 \leqslant j < k \leqslant N} {w_N(x_j-x_k)}.
\end{equation}
In order to model a short-range interaction, we will take 
\begin{equation}  \label{eq:ass-wN}
w_N(x)= N^{d\beta} w(N^\beta x).
\end{equation}
where $\beta>0$ and $w:\R^d\to \R$ are fixed. We put the factor $(N-1)^{-1}$ in front of the interaction potential in order to ensure that the interaction energy and the kinetic energy have the same order of magnitude.

We can think of the initial state $\Psi_N(0)$ as a ground state of a trapped system described by the Hamiltonian 
$$H_N^{V}=H_N+\sum_{j=1}^N V(x_i)$$
with $V$ an external trapping potential. When the trapping potential is turned off, $\Psi_N(0)$ is no longer a ground state of $H_N$, and its nontrivial evolution is described by $\Psi_N(t)$ in \eqref{eq:schrodingerdynamics}. 

Although the Schr\"odinger equation \eqref{eq:schrodingerdynamics} is linear, its complexity increases dramatically when $N$ becomes large. Therefore, for computational purposes, it is important to derive effective descriptions for collective behaviors of the quantum system.  To the leading order, complete Bose-Einstein condensation means $\Psi_N(0) \approx u(0)^{\otimes N}$ in an appropriate sense. If we believe that condensation is stable under the Schr\"odinger flow, namely $\Psi_N(t) \approx u(t)^{\otimes N}$, then, by formally substituting the ansatz $u(t)^{\otimes N}$ into the Schr\"odinger equation \eqref{eq:schrodingerdynamics}, we obtain the Hartree equation
\bq \label{eq:Hartree-equation}
i\partial_t u(t,x) =  \big(-\Delta_x +(w_N*|u(t,.)|^2)(x) -\mu_N(t)\big) u(t,x), 
\eq
where $\mu_N(t)$ is a phase factor. In the large $N$ limit, the Hartree evolution can be further approximated by the ($N$-independent) nonlinear Schr\"odinger equation (NLS)
\bq \label{eq:NLS}
i\partial_t \varphi(t,x)= (-\Delta_x + a|\varphi(t,x)|^2-\mu(t))\varphi(t,x).
\eq

Note that the meaning of the approximation $\Psi_N(t) \approx u(t)^{\otimes N}$ has to be understood properly as a convergence of one-body reduced density matrices. In particular, this approximation does not hold true in the norm topology of $\gH^N$, except for $w=0$ (non-interacting case).

The rigorous derivation of the Hartree and nonlinear Schr\"odinger equation has been the subject of a vast literature, initiated by Hepp \cite{Hepp-74}, Ginibre and Velo \cite{GinVel-79} and Spohn \cite{Spohn-80}. The results for  $0<\beta\le 1$ (dilute regime) has been studied by Erd\"os, Schlein and Yau \cite{ErdSchYau-07,ErdSchYau-09,ErdSchYau-10} for $d=3$, Adami, Golse and Teta \cite{AdaGolTet-07} for $d=1$ and Kirkpatrick, Schlein and Staffilani \cite{KirSchSta-11} for $d=2$ (see also \cite{JebLeoPic-16}). All these works deal with the defocusing case $w\ge 0$. 

In the focusing case $w\le 0$, the NLS is only well-posed globally for $d\le 2$, and it is natural to restrict to these low  dimensions. The derivation of the focusing NLS has been achieved very recently by Chen and H\"olmer  \cite{CheHol-16a,CheHol-16} by means of the BBGKY approach, and then by Jeblick and Pickl \cite{JebPic-17} by another approach. In these works, it is crucial to consider the $N$-body Hamiltonian $H_N^V$ with a trapping potential like $V(x)=|x|^s$, and restrict to $0<\beta<1$ when $d=1$ \cite{CheHol-16a} and $0<\beta<(s+1)/(s+2)$ when $d=2$ \cite{CheHol-16,JebPic-17}. The presence of the trap and the restriction on $\beta$ allow to use the stability of the second kind $H_N^V \ge -CN$ by Lewin, Nam and Rougerie \cite{LewNamRou-16b,LewNamRou-17}. This stability is important to control the (negative) interaction potential by the kinetic operator. 

In the present paper, we are interested in the norm approximation, which is much more precise than the convergence of density matrices. It requires to understand not only the condensate but also the fluctuations around the condensate. Motivated by rigorous results on ground states of trapped systems \cite{Seiringer-11,GreSei-13,LewNamSerSol-15,DerNap-13,NamSei-15,BBCS}, we will assume that the initial datum satisfies the norm approximation 
$$
\Psi_N(0) \approx \sum_{n=0}^N u(0)^{\otimes (N-n)} \otimes_s \psi_n(0).
$$
Bogoliubov's approximation \cite{Bogoliubov-47} suggests that 
\bq \label{eq:PsiNt-intro}
\Psi_N(t) \approx \sum_{n=0}^N u(t)^{\otimes (N-n)} \otimes_s \psi_n(t)
\eq
where $(\psi_n(t))_{n=0}^\infty$ is governed by an effective Hamiltonian in Fock space which is quadratic in creation and annihilation operators. 

The norm approximation \eqref{eq:PsiNt-intro} has been established by Lewin, Nam and Schlein \cite{LewNamSch-15} for $\beta=0$  (see \cite{MitPetPic-16} for a similar result). See \cite{NamNap-15,NamNap-17,BNNS} for extensions to $0<\beta<1$ with $w\ge 0$ in $d=3$ dimensions.  

The goal of the present paper is to extend the norm approximation \eqref{eq:PsiNt-intro} to lower dimensions, in both defocusing and focusing cases. The focusing case is more interesting as we have to rule out the instability of the system. Our method can treat all $\beta>0$ for $d=1$ and $0<\beta<1$ for $d=2$. 

So in particular, we can recover and improve the leading order results in \cite{CheHol-16a,CheHol-16,JebPic-17}. More precisely, we can deal with a larger (and much more natural) range of $\beta$ and do not need to add a trapping potential (although our method works with the presence of an external potential as well). To achieve our result we will use a new localization technique, which allows us to go beyond the stability regime established in \cite{LewNamRou-16b,LewNamRou-17}. Moreover, our main result not only deals with the condensate, but also provides a detailed description for the fluctuations.

An analogue of \eqref{eq:PsiNt-intro} related to the fluctuations around coherent states in Fock space has attracted many studies \cite{GriMacMar-11,GriMac-13,Kuz-15b,BocCenSch-17,GriMac-15}. All these works concentrate on the defocusing case in 3D. It is straightforward to apply our method to the Fock space setting to treat the focusing case in one- and two-dimensions.  

Finally, let us mention that in 3D, the focusing NLS may blow-up at finite time and it is interesting to derive effective equations locally in time. We refer to \cite{Pickl-10,NamNap-15,Chong-16} for relevant results in this direction. Our method should be also useful for this problem. We hope to come back this issue in the future.

The precise statement of our result is given in the next section.

\medskip

\noindent{\bf Acknowledgments.} We thank Christian Brennecke, Mathieu Lewin, Nicolas Rougerie, Benjamin Schlein and Robert Seiringer for helpful discussions. The support of the National Science Centre (NCN) project Nr. 2016/21/D/ST1/02430 is gratefully acknowledged (MN). 

\section{Main result} \label{sec:main-result}

First, we start with our assumption on the interaction potential. Recall from  \eqref{eq:ass-wN} that we are choosing
$$
w_N(x)= N^{d\beta} w(N^\beta x).
$$
We will always assume that 
\bq \label{eq:ass-wa}
w\in L^1(\R^d), \quad w(x)=w(-x) \in \R.
\eq

When $d=1$, there is no further assumption is needed (indeed, our proof can be even extended to a delta interaction $\pm\delta_0$). In $d=2$, we need an additional assumption that
\bq \label{eq:ass-wb}
w\in L^\infty(\R^d), \quad \int_{\R^2} |w_-| < a^*, \quad w_-=\max(-w,0).
\eq
Here $a^*$ is the optimal constant in the Gagliardo--Nirenberg inequality 
\begin{align}
\left( \int_{\R^2} |\nabla f|^2 \right) \left( \int_{\R^d} |f|^2 \right) \ge \frac{a^*}{2} \int_{\R^2} |f|^4, \quad \forall f\in H^1(\R^2).  \label{eq:GN-ineq}
\end{align}
Indeed, it is well-known \cite{Weinstein-83,Zhang-00,GuoSei-14,Maeda-10} that $a ^* := \norm{Q}_{L ^2} ^2 $ where $Q\in H ^1 (\R ^2)$ is the unique positive radial solution to 
\begin{equation}\label{eq:2D model prob}
-\Delta Q + Q - Q ^3 = 0.
\end{equation}
The condition \eqref{eq:ass-wb} is essential for the stability of the 2D system; see \cite{LewNamRou-16b} for a detailed discussion. This condition has been used in the derivation of the nonlinear Schr\"odinger equation for ground states of trapped system \cite{LewNamRou-16b,LewNamRou-17}, as well as for the dynamics \cite{CheHol-16,JebPic-17}.

\subsection*{\bf Hartree equation.} In our paper, the condensate $u(t)$ is governed by the Hartree equation \eqref{eq:Hartree-equation} with the phase
$$
\mu_N(t)=\frac12\iint_{\R^d\times\R^d}|u(t,x)|^2w_N(x-y)|u(t,y)|^2 \d x \d y
$$
which is chosen to ensure an energy compatibility (see \cite{LewNamSch-15} for further explanations). Here we have written $u(t)=u(t,.)$ and ignored the $N$-dependence in the notation for simplicity. The well-posedness of the Hartree equation will be discussed in Section \ref{sec:Hartree}. It is easy to see that when $N\to \infty$, $u(t)$ convergence to the solution of the cubic nonlinear Schr\"odinger equation \eqref{eq:NLS}. 

\subsection*{Bogoliubov equation.} To describe the fluctuations around the condensate, it is convenient to turn to the grand-canonical setting of Fock space 
$$ \cF(\gH)= \bigoplus_{n=0}^\infty \gH^n, \quad \gH^n = \bigotimes^n_{\rm sym} \gH, \quad \gH=L^2(\R^d)$$
where the number of particles can vary (and indeed, as we will soon see, the number of excitations will not be fixed). On Fock space, we define the creation and annihilation operators $a^*(f)$, $a(f)$, with $f\in \gH$, by
\begin{align*}
(a^* (f) \Psi )(x_1,\dots,x_{n+1})&= \frac{1}{\sqrt{n+1}} \sum_{j=1}^{n+1} f(x_j)\Psi(x_1,\dots,x_{j-1},x_{j+1},\dots, x_{n+1}), \\
 (a(f) \Psi )(x_1,\dots,x_{n-1}) &= \sqrt{n} \int \overline{f(x_n)}\Psi(x_1,\dots,x_n) \d x_n, \quad \forall \Psi \in \gH^n,\, \forall n. 
\end{align*}
These operators satisfy the canonical commutation relations (CCR)
$$ 
[a(f),a(g)]=[a^*(f),a^*(g)]=0,\quad [a(f), a^* (g)]= \langle f, g \rangle, \quad \forall f,g \in \gH.
$$
We can also define the operator-valued distributions $a_x^*$ and $a_x$, with $x\in \R^d$, by
$$
a^*(f)=\int_{\R^d}  f(x) a_x^* \d x, \quad a(f)=\int_{\R^d} \overline{f(x)} a_x \d x, \quad \forall f\in \gH,
$$
which satisfy  
$$[a^*_x,a^*_y]=[a_x,a_y]=0, \quad [a_x,a^*_y]=\delta(x-y), \quad \forall x,y\in \R^d.$$
The Hamiltonian $H_N$ can be extended to Fock space as
$$
H_N= \int a_x^* (-\Delta_x) a_x \d x + \frac{1}{2(N-1)}\iint w_N(x-y) a_x^* a_y^* a_x a_y \d x \d y.
$$
For every one-body operator $h$, we will use the short hand notation 
$$
\dGamma(h):= \int a_x^* h_x a_x \d x = 0\oplus \bigoplus_{n=0}^\infty \sum_{j=1}^n h_j. 
$$
In particular, $\cN=\dGamma(1)$ is called the number operator. 

We look for the norm approximation of the form 
\bq 
\Psi_N(t) \approx \sum_{n=0}^N u(t)^{\otimes (N-n)} \otimes_s \psi_n(t) := \sum_{n=0}^N \frac{(a^*(u(t)))^{N-n}}{\sqrt{(N-n)!}} \psi_n(t). \nonumber
\eq
Here the particles outside of the condensate are described by a unit vector 
$$\Phi(t)=(\psi_n(t))_{n=0}^\infty$$
in the excited Fock space 
$$
\cF_+(t)= \bigoplus_{n=0}^\infty \gH_+^n, \quad \gH_+^n=\bigotimes^n_{\rm sym} \gH_+(t), \quad \gH_+(t)= \{u(t)\}^\bot \subset L^2(\R^d).
$$
  As explained in \cite{LewNamSch-15}, Bogoliubov approximation suggests that the vector $\Phi(t)$ solves the equation
\bq \label{eq:Bogoliubov-equation}
i \partial_t \Phi(t) &=  \bH(t) \Phi(t)
\eq
where $\bH(t)$ is a ($N$-dependent)  quadratic Hamiltonian on the Fock space $\cF$ of the form 
$$
\bH(t)= \dGamma(h(t)) + \frac12\iint_{\R^d\times\R^d}\Big(K_2(t,x,y)a^*_x a^*_y +\overline{K_2(t,x,y)}a_x a_y\Big)\d x\,\d y.
$$
Here 
\begin{align*} h(t)=-\Delta+|u(t,\cdot)|^2\ast w_N -\mu_N(t) + Q(t) \widetilde{K}_1(t) Q(t), \quad K_2(t)=Q(t)\otimes Q(t)\widetilde{K}_2(t)
\end{align*}
where the kernel of the operator $\widetilde{K}_1(t)$ and the 2-body function $K_2(t)\in \gH^2$ are 
$$\widetilde{K}_1(t,x,y)=u(t,x)w_N(x-y)\overline{u(t,y)}, \quad \widetilde{K}_2(t,x,y)=u(t,x)w_N(x-y)u(t,y).$$

The well-posedness of \eqref{eq:Bogoliubov-equation} will be revisited in Section \ref{sec:Bogoliubov}. Note that $\Phi(t)$ belongs to $\cF_+(t)$, which is not obviously seen from the equation \eqref{eq:Bogoliubov-equation}.

\subsection*{Main result.} Now we are ready to state our main result.

\begin{theorem}[Norm approximation] \label{thm:main} Let $\beta>0$ when $d=1$ and $0<\beta<1$ when $d=2$. Assume that the interaction potential $w$ satisfies \eqref{eq:ass-wa}--\eqref{eq:ass-wb}. Let $u(t)$ satisfy the Hartree equation \eqref{eq:Hartree-equation} with $\|u(0)\|_{H^{d+2}(\R^d)}\le C$. Let $\Phi(t)=(\varphi_n(t))_{n=0}^\infty$ satisfy the Bogoliubov equation \eqref{eq:Bogoliubov-equation} with $\big\langle \Phi(0), \dGamma(1-\Delta) \Phi(0) \big\rangle\le C$. Consider the $N$-body Schr\"odinger evolution $\Psi_N(t)$ in \eqref{eq:schrodingerdynamics} with the initial state
\bq 
\Psi_N(0) = \sum_{n=0}^N u(0)^{\otimes (N-n)} \otimes_s \psi_n(0) . \nonumber
\eq
Take $\alpha=1/2$ when $d=1$ and $0<\alpha<(1-\beta)/3$ arbitrarily when $d=2$. Then for all $t>0$, there exists a constant $C_t>0$ independent of $N$ such that for all $N$ large,
\begin{align} \label{eq:thm-mainresult}
\Big\| \Psi_N(t) - \sum_{n=0}^N u(t)^{\otimes (N-n)} \otimes_s \psi_n(t) \Big\|_{\gH^N}^2 \le C_{t}   N^{-\alpha}.
\end{align}
\end{theorem}

\begin{remark} We have some immediate remarks concerning the result in Theorem \ref{thm:main}.
\begin{itemize}

\item[(1)] The initial state $\Psi_N(0)$ is not normalized, but its norm converges to $1$ very fast when $N\to \infty$ (as we will explain in the proof). Hence, we ignore a trivial normalization in the statement of Theorem \ref{thm:main}. 

\item[(2)] Our approach is quantitative and our result applies equally well for trapped systems. Moreover, the initial values $u(0)$ and $\Phi(0)$ can be chosen $N$-dependently, as soon as $\|u(0)\|_{H^{d+2}(\R^d)}$ and $\big\langle \Phi(0), \dGamma(1-\Delta) \Phi(0) \big\rangle$ grow slowly enough. We do not include these extensions to simplify the representation.

\item[(3)] When $d=2$, our result holds true if  \eqref{eq:ass-wb} is replaced by the weaker condition from \cite{LewNamRou-16b}
$$\inf_{u\in H^1(\R^2)}\left(\frac{\displaystyle\iint_{\R^2\times \R ^2}|u(x)|^2|u(y)|^2w(x-y)\,dx\,dy}{2\norm{u}_{L^2(\R^2)}^2\norm{\nabla u}_{L^2(\R^2)}^2}\right)>-1.$$
The latter condition is enough to ensure the well-posedness of the Hartree equation.

\item[(4)] A simplified formulation for Bogoliubov equation \eqref{eq:Bogoliubov-equation} can be given in terms of density matrices. Recall that for a vector $\Phi$ in Fock space, its one-body density matrices $\gamma_\Psi: \gH\to \gH$ and $\alpha_\Psi:\overline{\gH} \equiv \gH^* \to {\gH}$ can be defined by
$$
\left\langle {f,{\gamma _\Psi }g} \right\rangle  = \left\langle \Psi, {{a^*}(g)a(f)} \Psi \right\rangle,\quad \left\langle {{f}, \alpha _\Psi \overline{g} } \right\rangle  = \left\langle \Psi, {a(g)a(f)} \Psi\right\rangle, \quad \forall f,g \in \gH.
$$
As explained in \cite{NamNap-15}, if $\Phi(t)$ solves \eqref{eq:Bogoliubov-equation}, then $(\gamma(t),\alpha(t))=(\gamma_{\Phi(t)}, \alpha_{\Phi(t)})$ is the unique solution to the system
\bq \label{eq:linear-Bog-dm} 
\left\{
\begin{aligned}
i\partial_t \gamma &= h \gamma - \gamma h + K_2 \alpha - \alpha^* K_2^*, \\
i\partial_t \alpha &= h \alpha + \alpha h^{\rm T} + K_2  + K_2 \gamma^{\rm T} + \gamma K_2,\\
\gamma(0)&=\gamma_{\Phi(0)}, \quad \alpha(0)  = \alpha_{\Phi(0)}.
\end{aligned}
\right.
\eq
Note that \eqref{eq:linear-Bog-dm} is similar (but not identical) to the equations studied in the Fock space setting \cite{GriMac-13,Kuz-15b,BacBreCheFroSig-15}.  Reversely, if $\Phi(0)$ is a quasi-free state, then $\Phi(t)$ is a quasi-free state for all $t>0$, and in this case the equation \eqref{eq:linear-Bog-dm} is equivalent to the Bogoliubov equation \eqref{eq:Bogoliubov-equation}. 
\end{itemize}
\end{remark}

As we have mentioned in the introduction, the norm approximation is much more precise than the convergence of density matrices. Recall that the one-body density matrix of a $N$-body wave function $\Psi_N$ is obtained by taking the partial trace 
$$\gamma_{\Psi_N}^{(1)}=\Tr_{2\to N}|\Psi_N\rangle\langle \Psi_N|.$$
Equivalently, $\gamma_{\Psi_N}^{(1)}$ is a trace class operator on $L^2(\R^d)$ with kernel
$$
\gamma_{\Psi_N}^{(1)}(x,y)=\int \Psi_N(x,x_2,...,x_N) \overline{\Psi_N(y,x_2,...,x_N)} \d x_2 ... \d x_N.
$$
The following result is a direct consequence of Theorem \ref{thm:main}.

\begin{corollary}[Convergence of reduced density]\label{cor} Under the same conditions in Theorem \ref{thm:main}, we have the convergence in trace class 
$$
\lim_{N\to \infty} \gamma_{\Psi_N(t)}^{(1)}= |\varphi(t)\rangle \langle \varphi(t)|, \quad \forall t>0,
$$
where $\varphi(t)$ is the solution to the following cubic nonlinear Schr\"odinger equation, 
\bq \label{eq:NLS-0phase}
i\partial_t \varphi(t,x)= (-\Delta_x + a|\varphi(t,x)|^2)\varphi(t,x), \quad \varphi(0,x)=u(0,x), \quad a=\int w.
\eq
\end{corollary}

Note that \eqref{eq:NLS-0phase} is different from the equation \eqref{eq:NLS} in the introduction because there is no phase $\mu(t)$. However, the phase plays no role when we consider the projection $|\varphi(t)\rangle \langle \varphi(t)|$.

Corollary \ref{cor} recovers and improves the existing results on the leading order: \cite{KirSchSta-11,Chen-12} for defocusing 2D with $0<\beta<1$; \cite{CheHol-16a} for focusing 1D with $0<\beta<1$; \cite{CheHol-16,JebPic-17} for focusing 2D with $0<\beta<(s+1)/(s+2)$ with a trapping potential like $|x|^s$. Here our range of $\beta$ is larger, i.e. $\beta>0$ for $d=1$ and $0<\beta<1$ for $d=2$, which is optimal to some extent. Moreover, we do not have to restrict to trapped systems, which is consistent with our interpretation that $\Psi_{N}(t)$ is the evolution of a ground state of a trapped system when the trap is turned off (but our proof works equally well with trapped systems). 

To improve the range of $\beta$ and to remove the trap, we will not rely on the stability of the second kind $H_N \ge -CN$ in \cite{LewNamRou-16b,LewNamRou-17}. Indeed, thanks to a new localization method on the number of particles, we will only need a weaker version of the stability on the sector of very few particles, which is much easier to achieve. This weaker stability is enough for our purpose because the fluctuations around the condensate involve only very few particles (most of particles are already in the condensate mode).  

\medskip

\noindent{\bf Organization of the paper.} The paper is organized as follows. We will always focus on the more difficult case $d=2$, and only explain the necessary changes for $d=1$ at the end. We will revise the well-posedness of the Hartree equation \eqref{eq:Hartree-equation} in Section \ref{sec:Hartree} and the Bogoliubov equation \eqref{eq:Bogoliubov-equation} Section \ref{sec:Bogoliubov}. In Section \ref{sec:Bog-app}, we reformulate the Bogoliubov approximation using a unitary transformation from $\gH^N$ to a truncated Fock space, following ideas in \cite{LewNamSerSol-15,LewNamSch-15}. Then we provide several estimates which are useful to implement Bogoliubov's approximation. In Section \ref{sec:loc} we explain the localization method. Then we prove Theorem \ref{thm:main} is presented in Section \ref{sec:main-proof}, for $d=2$. All the changes needed to prove  Theorem \ref{thm:main} for $d=1$ are explained in Section \ref{sec:last-d1}.

\section{Hartree dynamics}\label{sec:Hartree}

In this section, we discuss the well-posedness and provide various estimates for Hartree equation \eqref{eq:Hartree-equation}.

Our proof will require bounds on the solutions of the Hartree equation \eqref{eq:Hartree-equation}. Under our assumptions on the nonlinearity, it is well-known that the equation is globally well-posed in $H^1(\R^d)$ (see, e.g., \cite[Cor. 6.1.2]{Caz-03}). However, since the potential $w_N$ depends on $N$, it is not clear if the norm $\|u(t,.)\|_{H^1}$ is bounded uniformly in $N$. The same question applies to other norms we will use in the proof. We are going to prove that it is indeed the case. Here we will consider the case $d=2$ in detail. Remarks about the corresponding results for $d=1$ will be given in Section \ref{sec:last-d1}.

\begin{lemma}\label{lem:Hartree} Let $d=2$. Assume $w\in L^1(\R^2)\cap L^\infty(\R^2)$ and $\int w_-< a^*$. For every $u_0\in H^4(\R^2)$ with $\|u_0\|_{L^2}=1$, equation \eqref{eq:Hartree} has a unique solution $u(t,\cdot)$ in $H^4(\R^2)$ and we have for all $t>0$ the bounds
\begin{align*}
&\|u(t,\cdot)\|_{H^1(\R^2)} \le C, \quad  \|u(t,\cdot)\|_{H^2(\R^2)} \le C\exp(C\exp(Ct)), \\
& \|u(t,\cdot)\|_{L^\infty(\R^2)} \leq C\exp{(Ct)} ,\quad  \|\partial_t u(t,\cdot)\|_{L^\infty(\R^2)}\leq C\exp(\exp (C\exp(Ct))).
\end{align*}
for a constant $C>0$ independent of $t$ and $N$. 
\end{lemma}

\begin{proof} For convenience, we will work with the equation
\bq \label{eq:Hartree}
i\partial_t u = -\Delta u + (w_N*|u|^2) u, \quad u(0,x)=u_0(x), \quad x\in \R^d.
\eq
To go from equation \eqref{eq:Hartree} to equation \eqref{eq:Hartree-equation} it is enough consider a gauge transformation 
$$ u \mapsto e^{-i\int_0^t \mu_N (s)ds} u$$
with 
$$
\mu_N(t)=\frac12\iint_{\R^d\times\R^d}|u(t,x)|^2w_N(x-y)|u(t,y)|^2 \d x \d y.
$$
The  $N$-dependence in the desired bounds is not affected by this change. 

\textit{Step 1.} First, the bound in $H^1$ follows from the energy conservation
$$
\|\nabla u(t,.)\|_{L^2}^2 +\frac{1}{2} \iint |u(t,x)|^2 w_N(x-y) |u(t,y)|^2 \d x \d y =C,
$$ 
the simple estimate
\begin{align*}
&\iint |u(t,x)|^2 w_N(x-y) |u(t,y)|^2 \d x \d y \\
&\ge - \iint   |u(t,x)|^2 |[w_N(x-y)]_- | |u(t,y)|^2\d x \d y \\
&\ge - \iint \frac{|u(t,x)|^4+|u(t,y)|^4}{2} |[w_N(x-y)]_- | \d x \d y \\
&= -\|u(t,.)\|_{L^4}^4 \int |w_-|
\end{align*}
and the bounds \eqref{eq:ass-wb} and \eqref{eq:GN-ineq}.\\

\textit{Step 2.} Next, to prove the $H^2$ bound, we use Duhamel's formula
$$
u(t,x)= e^{-it\Delta} u_0(x) + \int_0^t e^{-i\Delta(t-s)} G(s,x)  \d s
$$
with
$$ 
G(s,\cdot)= (w_N*|u(s,\cdot)|^2)u(s,\cdot) .
$$ 
It follows that
$$
\|\Delta u(t,\cdot) \|_{L^2} \le \|\Delta u_0\|_{L^2} + \int_0^t \|\Delta G(s,\cdot)\|_{L^2} \d s.
$$
We compute
\begin{equation}
\begin{aligned}
\Delta G &= \Delta \Big[(w_N*|u|^2)u \Big]   \\
&= (w_N*|u|^2) \Delta u +  \sum_{i=1}^2 2 \Big[w_N* (\partial_{x_i} (|u|^2) ) \Big] \partial_{x_i} u + \Big[w_N* (\Delta (|u|^2))\Big] u. \label{eq:DeltaG2d}
\end{aligned}
\end{equation}
For the first term we easily find
$$\|(w_N*|u|^2) \Delta u \|_{L^2}\leq C\|w_N\|_{L^1}\|u(s,\cdot)\|_{L^\infty}^2\|\Delta u\|_{L^2}.$$
Since $\partial_{x_i} (|u|^2) =\bar{u}\partial_{x_i}u+u\partial_{x_i}\bar{u}$ in the second term of \eqref{eq:DeltaG2d} we need to bound
\begin{align*}
\|&\int w_N(y)\bar{u}(\cdot-y)\partial_{x_i}\bar{u}(\cdot-y)\d y \,\partial_{x_i}u(\cdot)\|_{L^2}\leq  \\
 & C\|u\|_{L^\infty} \big\|\int |w_N(y)\left(|\partial_{x_i}\bar{u}(\cdot-y)|^2+|\partial_{x_i}u(\cdot)|^2\right)\big \|_{L^2}\leq C\|u\|_{L^\infty} \|w_N\|_{L^1}\||\nabla u|^2\|_{L^2}.
\end{align*}
To bound the last term we use the the two dimensional Sobolev inequality 
$$\|\nabla u \|_4^2 \leq C\left(\|u\|_{H^1}^2+\|\Delta u\|_{L^2}^2\right).$$
Treating other terms in \eqref{eq:DeltaG2d} in a similar way we obtain
$$
\|\Delta G(s,\cdot)\|_{L^2} \le C(1+ \|\Delta u(s,\cdot)\|_{L^2}) (1+\|u(s,\cdot)\|_{L^\infty}^2).
$$
Thus we deduce that
$$
\|u(t,\cdot) \|_{H^2} \le C + C \int_0^t  \|u(s,\cdot)\|_{H^2} (1+ \|u(s,\cdot)\|_{L^\infty}^2) \d s .
$$

We can now use the two dimensional Brezis--Gallouet--Wainger (or log-Sobolev) inequality  \cite{BreGal-80,BreWai-80}
\begin{align}
\|v(s,\cdot)\|_{L^\infty}^2 \le C(1+\log (1+\|v(s,\cdot)\|_{H^2})) \label{eq:Brezis-Gallouet-Waigner}
\end{align}
which holds true for functions $v$ with $\|v\|_{H^1}=1$. By the assumption $\|u_0\|_{L^2}=1$ and mass preservation it follows that $\|u(t,\cdot)\|_{H^1}\geq 1$. Using this and $\|u(t,\cdot)\|_{H^1}\leq C$ from step 1 we deduce from the Brezis--Gallouet--Waigner inequality for $v=u/\|u\|_{H^1}$ that
$$
\|u(t,\cdot) \|_{H^2} \le C + C \int_0^t  \|u(s,\cdot)\|_{H^2} (1+\log (1+\|u(s,\cdot)\|_{H^2})) \d s .
$$
Denote
$$
F(t)= C + C \int_0^t  \|u(s,\cdot)\|_{H^2} (1+\log (1+\|u(s,\cdot)\|_{H^2})) \d s.
$$
Then
$$
F'(t) = C \|u(t,\cdot)\|_{H^2} (1+\log (1+\|u(t,\cdot)\|_{H^2})) \le C F(t) [1+\log (1+F(t))]
$$
which implies that
$$
\frac{d}{dt} \log (1+ \log (1+F(t))) \le C.
$$
This allows us to conclude that
$$
\|u(t,.)\|_{H^2} \le \exp(C \exp(Ct))
$$
which using \eqref{eq:Brezis-Gallouet-Waigner} immediately implies the bound on $\|u(t,\cdot)\|_{L^\infty}$.

\textit{Step 3.} Let us now prove the last bound. We will again use use the Brezis--Gallouet--Waigner inequality, this time for $\partial_t u$. It reads
\begin{align}
\|\partial_t u(t,\cdot)\|_{L^\infty}^2 \le C_0(t)(1+\log (1+\|\partial_t u(t,\cdot)\|_{H^2})) \label{eq:BGWineq-time_der}
\end{align}
where
$$C_0(t)=\max\{\|\partial_t u(t,\cdot)\|_{H^1}^2,1\}.$$ 
Since $\|v\|_{H^1}\leq \|v\|_{H^2}$ it is clear that we need to obtain a bound on $\|\partial_t u(t,\cdot)\|_{H^2}$.
From the Hartree \eqref{eq:Hartree} equation we get 
$$\|\partial_t u(t,\cdot)\|_{H^2}\leq \|\Delta u(t,\cdot)\|_{H^2}+\int_0^t \|G(s,\cdot)\|_{H^2}.$$
The norm equivalence 
$$c\left( \|f\|_{L^2}^2+\|\Delta f\|_{L^2}^2\right)^{1/2}\leq \|f\|_{H^2}\leq C\left( \|f\|_{L^2}^2+\|\Delta f\|_{L^2}^2\right)^{1/2}$$
thus implies
$$\|\partial_t u(t,\cdot)\|_{H^2}\leq C\left(\|\Delta u\|_{L^2}+\|\Delta^2 u\|_{L^2}\right)+\int_0^t C\left(\|G(s,\cdot)\|_{L^2}+\|\Delta G(s,\cdot)\|_{L^2}\right)ds.$$
Bounds obtained in the previous steps thus lead to 
\begin{equation}
\|\partial_t u(t,\cdot)\|_{H^2}\leq C(1+t)\exp(C\exp(Ct))+C\| u(t,\cdot)\|_{H^4}.\label{eq:partial_t_u_inftybound}
\end{equation}
By Duhamel's formula, for any integer $k$, we have
\begin{equation}
\|u(t,\cdot)\|_{H^k}\leq  \|u_0(\cdot)\|_{H^k} + \int_0^t \|G(s,\cdot)\|_{H^k} \d s. \label{eq:Duhamel_H^k}
\end{equation}
We shall first get a bound on $\|u(t,\cdot)\|_{H^3}$ and therefore we first look at $\|G(t,\cdot)\|_{H^3}$. To this end we notice that by the multi-index Leibniz formula
we have
\begin{align}
& \|G(t,\cdot)\|_{H^3}=\sum_{|\ell|\leq 3}\|D^{\ell}G\|_{L^2}  \label{eq:H^3normG} \\
& \leq \sum_{|\ell| \leq 3}\sum_{k\leq \ell}\sum_{m
\leq k} {{3}\choose{k}}{{k}\choose{m}}\| \int w_N(y)\partial^{m}u(x-y)\partial^{k-m}\bar{u}(x-y)\d y \partial^{\ell-k}u(x)\|_{L^2}. \nn
\end{align}
We will derive a Gr\"onwall type inequality. Since our goal is to obtain a bound that is independent of $N$, in our bounds we will need to extract the $L^1$ norm of $w_N$; otherwise the $N$-dependence will appear. Thus, to do the $\d y$ integration we need to  
use the $L^\infty$ bounds on the derivatives of $u$. By \eqref{eq:Brezis-Gallouet-Waigner} we have
\begin{align}
\|\partial^i u(t,\cdot)\|_{L^\infty}^2 \le & C_0(t)(1+\log (1+\|\partial^i u(t,\cdot)\|_{H^2})) \nn \\
\leq & C_0(t)(1+\log (1+\| u(t,\cdot)\|_{H^{2+|i|}}))   \label{eq:BGWineq-high-der}
\end{align}
where
$$C_0(t)=C\max\{\|\partial^i u(t,\cdot)\|_{H^1}^2,1\}\leq C\max\{\|u(t,\cdot)\|_{H^{1+|i|}}^2,1\}.$$ 
Let us now look at different terms in the sum in \eqref{eq:H^3normG}. The ingredients for the case $|\ell|=0$ can be trivially bounded by
\begin{align}
\left| |\ell|=0 \,\,\text{term} \right|\leq C\|u\|_{L^\infty}^2\|w_N\|_1 \|u\|_{L^2}\leq C\exp(Ct). \label{eq:term_l=0}
\end{align} 
For $|\ell|=1$ we obtain the bound
\begin{align}
\left| |\ell|=1 \,\,\text{term} \right|\leq C\|u\|_{L^\infty}^2\|w_N\|_1 \|u\|_{H^1}\leq C\exp(Ct). \label{eq:term_l=1}
\end{align} 
For the $|\ell|=2$ term we have two possibilities. Either both derivatives hit one function or they distribute among two functions. This leads to
\begin{align}
\left| |\ell|=2 \,\,\text{term} \right|& \leq C\|u\|_{L^\infty}^2\|w_N\|_1 \|u\|_{H^2}+ C\|u\|_{L^\infty}\|\partial_i u\|_{L^\infty}\|w_N\|_1 \|u\|_{H^1} \nn \\
& \leq C(1+\exp(C\exp(Ct)))(1+\log(1+\|u\|_{H^3})) \label{eq:term_l=2}
\end{align} 
where we used \eqref{eq:BGWineq-high-der}. It remains to bound the term corresponding to $|\ell|=3$. We again notice that if all derivatives hit one function then we get a bound  $C\exp(Ct)\|u\|_{H^3}$. In the case when two derivatives hit one function and and one derivative hits one of the others we get a bound as in \eqref{eq:term_l=2}. The same bound works for the case when one derivative hits each function. Altogether we arrive at the bound for the $|\ell|=3$ term of the form
 \begin{align}
\left| |\ell|=3 \,\,\text{term} \right| \leq C(1+\exp(C\exp(Ct)))(1+\|u\|_{H^3}). \label{eq:term_l=3}
\end{align}
Collecting the bounds \eqref{eq:term_l=0}-\eqref{eq:term_l=3} and inserting them into \eqref{eq:H^3normG} yields 
$$\|G(t,\cdot)\|_{H^3}\leq C(1+\exp(C\exp(Ct)))(1+\|u\|_{H^3}).$$
Inserting this into \eqref{eq:Duhamel_H^k} implies by Gr\"onwall inequality that
\begin{equation}
\|u(t,\cdot)\|_{H^3}\leq C\exp(C\exp(C\exp(Ct))). \label{eq:H^3bound_u}
\end{equation}
With this result we are finally able to proceed to the last estimate, that is the bound on $\|u(t,\cdot)\|_{H^4}$. In this case the analysis proceeds in the same way. We again write out the $H^4$ counterpart of \eqref{eq:H^3normG} with $|\ell|$ now being not bigger than four. The analysis above shows that all terms in the sum which correspond to $|\ell|\leq 3$ can be bounded in terms of $C\exp(C\exp(C\exp(Ct)))$. Let us now look at different terms corresponding to $|\ell|=4$. If all derivatives hit one function then obviously we bound it by $C\exp(Ct) \|u\|_{H^4}$. If two derivatives hit one function and two another one, then using \eqref{eq:BGWineq-high-der} with $|i|=2$ we obtain a bound of the form $C\exp(C\exp(C\exp(Ct)))(1+\log(1+\|u\|_{H^4}))$. Other terms can be bounded independently of $\|u\|_{H^4}$ by $C\exp(C\exp(C\exp(Ct)))$.  Again, using Gr\"onwall, we arrive at he bound 
\begin{equation}
\|u(t,\cdot)\|_{H^4}\leq C(\exp(C\exp(C\exp(C\exp(Ct)))). \label{eq:H^4bound_u}
\end{equation}
Inserting \eqref{eq:H^4bound_u} into \eqref{eq:partial_t_u_inftybound} and this into \eqref{eq:BGWineq-time_der}, we see that finally
$$\|\partial_t u(t,\cdot)\|_{L^\infty}^2\leq C\exp(\exp (C\exp(Ct))).$$
This ends the proof.
\end{proof}

From now on, we will often omit the explicit time-dependence on the constant and replace it by a general notation $C_t$, for simplicity. Also, we will focus on $d=2$.

\section{Bogoliubov dynamics}\label{sec:Bogoliubov}

The main result of this section is 

\begin{lemma} \label{lem:bH-kinetic} Assume that $\Phi(0)$ satisfies $\langle \Phi(0), \dGamma(1-\Delta) \Phi(0)\rangle \le C$. Then the Bogoliubov equation \eqref{eq:Bogoliubov-equation} has a unique global solution 
$$\Phi \in C([0,\infty), \cF(\gH)) \cap L^\infty_{\rm loc} ((0,\infty), \mathcal{Q}(\dGamma (1-\Delta))).$$
Moreover, 
$$
\big \langle \Phi(t), \dGamma(1-\Delta)  \Phi(t) \big\rangle \le  C_{t,\eps} N^\eps. 
$$
\end{lemma}

The well-posedness in Lemma \ref{lem:bH-kinetic} follows from \cite[Theorem 7]{LewNamSch-15}. The new result is the kinetic estimate. This follows from the following quadratic form estimates on $\cF(\gH)$.

\begin{lemma} \label{lem:bHt-dbHt} For every $\eps>0$ and $\eta>0$, 
\begin{align*}
\pm \Big( \bH(t) + \dGamma(\Delta) \Big) &\le \eta \dGamma(1-\Delta) + C_{t,\eps}(1+\eta^{-1})(\cN + N^\eps),\\
\pm \partial_t \bH(t) &\le \eta \dGamma(1-\Delta) +   C_{t,\eps}(1+\eta^{-1})(\cN + N^\eps),\\
\pm i[\bH(t),\cN] &\le \eta \dGamma(1-\Delta) +   C_{t,\eps}(1+\eta^{-1})(\cN + N^\eps),
\end{align*}
\end{lemma}

Let us assume Lemma \ref{lem:bHt-dbHt} for the moment and provide

\begin{proof}[Proof of Lemma \ref{lem:bH-kinetic}] We will use Gronwall's argument. By Lemma \ref{lem:bHt-dbHt}, we have 
\begin{align*}
A(t):= \bH(t) + C_{t,\eps} (\cN + N^\eps)  \ge  \frac{1}{2}\dGamma(1-\Delta),
\end{align*}
for $C_{t,\eps}$ sufficiently large and, by the equation for $\Phi(t)$,
\begin{align*}
\frac{d}{dt} \big\langle \Phi(t), A(t) \Phi(t)  \big\rangle &= \big\langle \Phi(t), \partial_t A(t)  \Phi(t)  \big\rangle + \big\langle \Phi(t), i[\bH(t),A(t)]  \Phi(t)  \big\rangle \\
&= \big\langle \Phi(t), \partial_t (\bH(t)+\partial_t C_{t,\eps}(\cN + N^\eps) ) \Phi(t)  \big\rangle + \big\langle \Phi(t), i[\bH(t),\cN]  \Phi(t)  \big\rangle\\
&\le C_{t,\eps} \big\langle \Phi(t), A(t)  \Phi(t) \rangle.
\end{align*}
Thus
$$
\big\langle \Phi(t), A(t) \Phi(t)  \big\rangle \le e^{C{t,\eps}}\big\langle \Phi(0), A(0) \Phi(0)  \big\rangle.
$$
The left side is bounded from below by $\frac{1}{2}\big\langle \Phi(t), \dGamma(1-\Delta) \Phi(t)  \big\rangle$. The right side can be bounded by using 
$$
A(0)\le C_\eps( \dGamma(1-\Delta) + N^\eps).
$$
We thus conclude that
$$
\big\langle \Phi(t), \dGamma(1-\Delta) \Phi(t)  \big\rangle \le C_\eps e^{C{t,\eps}} \Big( \big\langle \Phi(0), \dGamma(1-\Delta) \Phi(0)  \big\rangle \Big) + N^\eps \Big).
$$
This ends the proof.
\end{proof}

Now we turn to the proof of Lemma \ref{lem:bHt-dbHt}. We will need two preliminary results. The first is a lower bound on general quadratic Hamiltonians, taken from \cite[Lemma 9]{NamNapSol-16}.

\begin{lemma} \label{lem:Bog-GSE} Let $H>0$ be a self-adjoint operator on $\gH$. Let $K:\overline{\gH}\equiv \gH^*\to \gH$ be an operator with kernel $K(x,y) \in \gH^2$. Assume that $K H^{-1} K^* \le H$ and that $H^{-1/2}K$ is Hilbert-Schmidt. Then 
$$ \dGamma(H) + \frac{1}{2} \iint \Big( K(x,y) a_x^* a_y^* + \overline{K(x,y)}a_x a_y \Big) \d x \d y \ge -\frac{1}{2} \| H^{-1/2} K\|_{\rm HS}^2.$$
\end{lemma}

The second is the following kernel estimate.
\begin{lemma}\label{lem:Sobolev-inverse-K2} For all $\eps>0$ we have 
\begin{align*}
\|(1-\Delta_x)^{-1/2} K_2(t, \cdot, \cdot)\|^2_{L^2} + \|(1-\Delta_x)^{-1/2} \partial_t K_2(t, \cdot, \cdot)\|^2_{L^2} \le C_{t,\eps} N^{\eps}.
\end{align*}
\end{lemma}

\begin{proof}[Proof of Lemma \ref{lem:Sobolev-inverse-K2}] We will focus on $\partial_t K_2(t)$ as $K_2(t)$ can be treated similarly. By the definition  
$$ K_2(t,\cdot ,\cdot)=Q(t)\otimes Q(t) \widetilde K_2(t,\cdot,\cdot), \quad \widetilde K_2(t,x,y)=u(t,x)w_N(x-y) u(t,y).$$
we have
$$ \partial_t K_2(t)= \partial_t Q(t) \otimes Q(t) \widetilde K_2(t) + Q(t) \otimes \partial_t Q(t) \widetilde K_2(t)+Q(t)\otimes Q(t) \partial_t \widetilde K_2(t).$$

First, let us compare $\partial_t K_2(t)$ with $\partial_t \w K_2(t)$. Since $\partial_t Q(t)=-| \partial_t u(t)\rangle \langle u(t)| -|u(t)\rangle \langle \partial_t u(t)|$, we have
\begin{align*}
&\|\partial_t Q(t) \otimes Q(t) \widetilde K_2(t,\cdot,\cdot)\|_{L^2} \le \| (\partial_t Q(t) \otimes 1) \widetilde K_2(t,\cdot,\cdot)\|_{L^2} \\
& \le \| (| \partial_t u(t)\rangle \langle u(t)| \otimes 1 )  \widetilde K_2(t,\cdot,\cdot)\|_{L^2} + \| (|  u(t)\rangle \langle \partial_t u(t)|  \otimes 1 ) \widetilde K_2(t,\cdot,\cdot)\|_{L^2}
\end{align*}
For the first term, it is straightforward to see that
\begin{align} 
&\left\| (|\partial_t u\rangle \langle u|  \otimes 1) \widetilde K_2(t, \cdot, \cdot) \right\|_{L^2}^2 \nn\\
&=\iint  \left| \int \overline{u(t,z)}  u(t,z) w_N(z-y) u(t,y) \d z \right|^2 |\partial_t u(t,x)|^2 \d x \d y \label{eq:partialtK2bound} \\
& \le \|u(t,\cdot)\|_{L^\infty}^4 \|w_N\|_{L^1}^2 \|u(t,\cdot)\|_{L^2}^2   \|\partial_t u(t,\cdot)\|_{L^2}^2  \le C_t. \nn
\end{align}
The second term is also bounded by the same way. Thus we find that
$$
\|\partial_t Q(t) \otimes Q(t) \widetilde K_2(t,\cdot,\cdot)\|_{L^2} \le C_t.
$$
By the same argument, we also obtain
$$
 \| Q(t) \otimes \partial_t Q(t) \widetilde K_2(t,\cdot,\cdot)\|_{L^2} \le C_t
 $$
and  
$$\| (1-Q(t)\otimes Q(t)) \partial_t \widetilde K_2(t,\cdot,\cdot)\|_{L^2} \le C_t.$$
In summary,
$$
\|  \partial_t K_2(t,\cdot,\cdot) -  \partial_t \w K_2(t,\cdot,\cdot)\|_{L^2} \le C_t.
$$
Since $(1-\Delta_x)^{-1/2} \le 1$ on $L^2$, we obtain 
$$
\|  (1-\Delta_x)^{-1/2}
 \partial_t K_2(t,\cdot,\cdot) -  (1-\Delta_x)^{-1/2}
 \partial_t \w K_2(t,\cdot,\cdot)\|_{L^2} \le C_t. $$

It remains to bound $(1-\Delta_x)^{-1/2}
 \partial_t \w K_2(t,\cdot,\cdot)$. We have
 $$
 \partial_t \w K_2(t,\cdot,\cdot) = \partial_t u(t,x) w_N(x-y) u(t,y) +  u(t,x) w_N(x-y) \partial_t u(t,y)
 $$
  and it suffices to consider $f(t,x,y)=\partial_t u(t,x) w_N(x-y) u(t,y)$ (the other term is similar). Inspired by the idea in \cite{GriMac-13}, we compute the Fourier transform:
\begin{align*}
\widehat{f}(t,p,q) &= \iint u(t,x)w_N(x-y)(\partial_t u)(t,y) e^{-2\pi i (p\cdot x + q\cdot y)} \d x \d y \\
&=\iint u(t,y+z) w_N(z) (\partial_t u)(t,y) e^{-2\pi i (p\cdot (y+z) + q\cdot y)} \d z\d y \\
&=\int w_N(z) \widehat{(u_z \partial_t u)}(t,p+q) e^{-2\pi ip\cdot z}\d z
\end{align*}
where $u_z(t,\cdot):=u(t,z+\cdot).$ By the Cauchy-Schwarz inequality,
\begin{align*}
\left|\widehat{f}(t,p,q) \right|^2 \le \|w_N\|_{L^1} \int |w_N(z)| \cdot| \widehat{(u_z \partial_t u)}(t,p+q)|^2 \d z
\end{align*}
and by Lemma \ref{lem:Hartree}, 
\begin{align*}
\int |\widehat{(u_z \partial_t u)}(t,p+q)|^2 \d q & = \| (u_z \partial_t u)(t,\cdot) \|_{L^2}^2 \le C_t.
\end{align*}
Thus by Plancherel's Theorem, for all $\eps>0$ we have
\begin{align*}
& \|(1-\Delta_x)^{-1/2-\eps} f(t, \cdot, \cdot)\|^2_{L^2} = \iint (1+|2\pi p|^2)^{-1-2\eps}\left|\widehat{f}(t,p,q) \right|^2 \d p \d q \\
&\le \|w_N\|_{L^1}  \iiint (1+|2\pi p|^2)^{-1-2\eps} |w_N(z)| \cdot |\widehat{(u_z \partial_t u)}(t,p+q)|^2 \d p \d q \d z\\
&\le C_{t,\eps}.
\end{align*}
Here we have used the facts that $\|w_N\|_{L^1}=\|w\|_{L^1}$ and 
\begin{equation*}
\int  (1+|2\pi p|^2)^{-1-2\eps}  \d p \le C_\eps <\infty.
\end{equation*}
Moreover, by Lemma \ref{lem:Hartree} we have the simple estimate
$$
\|f(t, \cdot, \cdot)\|^2_{L^2} \le C\|u(t,\cdot)\|_{L^\infty}^2 \|\partial_t u(t,\cdot)\|_{L^2}^2 \|w_N\|_{L^2}^2 \le C_t N^{2\beta}.
$$
By interpolation, we thus obtain 
$$
 \|(1-\Delta_x)^{-1/2-\eps} f(t, \cdot, \cdot)\|^2_{L^2} \le C_{t,\eps} N^{\eps}, \quad \forall \eps>0.
$$
This ends the proof.
\end{proof}

Now we are ready to give

\begin{proof}[Proof of Lemma \ref{lem:bHt-dbHt}] We have
$$
\bH(t) + \dGamma(\Delta)= \dGamma(h+\Delta) + \frac{1}{2} \iint \Big( K_2(t,x,y) a_x^* a_y^* + \overline{K_2(t,x,y)}a_x a_y \Big) \d x \d y.
$$
By Lemma \ref{lem:Hartree}, it is straightforward to see that 
$$\| h+\Delta\| = \Big\| |u(t,\cdot)|^2\ast w_N -\mu_N(t) + Q(t) \widetilde K_1(t) Q(t) \Big\| \le C_t,$$
and hence
$$
\pm \dGamma(h+\Delta)\le C_t\cN.
$$
Now we consider the paring term. Note that $\|K_2\|\le C_t$. We apply Lemma \ref{lem:Bog-GSE} with $H=\eta(1-\Delta)+\eta^{-1}\|K_2\|^2$, $\eta>0$, $K=\pm K_2$ and then use Lemma \ref{lem:Sobolev-inverse-K2}. We get 
\begin{align*}
&\eta \dGamma(1-\Delta) + \eta^{-1} \cN \pm \frac{1}{2} \iint \Big( K_2(t,x,y) a_x^* a_y^* + \overline{K_2(t,x,y)}a_x a_y \Big) \d x \d y \\
&\ge -\frac{1}{2} \eta^{-1} \|(1-\Delta)^{-1/2}K_2\|_{\rm HS}^2 \ge -C_{t,\eps} \eta^{-1} N^\eps, \quad \forall \eps>0.
\end{align*}
Thus
$$
\pm \Big( \bH(t) + \dGamma(\Delta)  \Big) \le \eta \dGamma(1-\Delta) + C_t (1+\eta^{-1})\cN +  C_{t,\eps} \eta^{-1} N^\eps
$$
for all $\eta>0, \eps>0$.

The bound for $\partial_t \bH(t)$ is obtained by the same way and we omit the details. Moreover, it is straightforward to see that 
\begin{align*}
i[\bH(t),\cN] = - \iint \Big( iK_2(t,x,y) a_x^* a_y^* + \overline{iK_2(t,x,y)}a_x a_y \Big) \d x \d y
\end{align*}
and the bound for $i[\bH(t),\cN]$ also follows from the above argument. This completes the proof.
\end{proof}

\section{Bogoliubov's approximation}\label{sec:Bog-app}

As in ~\cite[Sec. 2.3]{LewNamSerSol-15}, for every normalized vector $\Psi\in \gH^N$ we can write uniquely as 
\begin{equation*}
\Psi=\sum_{n=0}^N u(t)^{\otimes (N-n)} \otimes_s \psi_n  = \sum_{n=0}^N \frac{(a^*(u(t)))^{N-n}}{\sqrt{(N-n)!}} \psi_n
\end{equation*}
with $\psi_n \in \gH_+(t)^{n}$. This gives rise the  unitary operator
\begin{equation*}
\begin{array}{cccl}
U_{N}(t): & \gH^N & \to & \displaystyle \cF_+^{\le N}(t):=\bigoplus_{n=0}^N \gH_+(t)^n \\[0.3cm]
 & \Psi & \mapsto & \psi_0\oplus \psi_1 \oplus\cdots \oplus \psi_N.
\end{array}
\end{equation*}
Thus
$$\Phi_N(t):=U_N(t) \Psi_N(t)$$
describes the fluctuations around the condensate $u(t)$. 

As proved in \cite{LewNamSch-15}, $\Phi_N(t)$ belongs to $\cF_+^{\le N}(t)$ and  solves the equation
\bq \label{eq:eq-PhiNt}
\left\{
\begin{aligned}
i \partial_t \Phi_N(t)  &=  \cG_N (t)   \Phi_N(t), \\
 \Phi_N(0) & = \1^{\le N} \Phi(0).
\end{aligned}
\right.
\eq
Here $\1^{\le m}=\1(\cN\le m)$ is the projection onto the truncated Fock space
$$\cF^{\le m}=\mathbb{C} \oplus \gH \oplus \cdots \oplus \gH^m$$
and 
$$
\cG_N (t)=  \1^{\le N} \Big[ \bH(t) +\cE_{N}(t)  \Big] \1^{\le N}
$$
with 
\begin{equation} \label{eq:E-error-R}
\cE_{N}(t)=\frac{1}{2}\sum_{j=0}^4 ( R_{j} + R_j^*),
\end{equation}
\begin{align*}
R_{0}&=R_0^*= \d\Gamma(Q(t)[w_N*|u(t)|^2+ \widetilde{K}_1(t) -\mu_N(t)]Q(t))\frac{1-\cN}{N-1},\\
R_{1}&=-2\frac{\cN\sqrt{N-\cN}}{N-1} a(Q(t)[w_N*|u(t)|^2]u(t)),\\
R_{2}&= \iint  K_2(t,x,y) a^*_x a^*_y \d x \d y \left(\frac{\sqrt{(N-\cN)(N-\cN-1)}}{N-1}-1\right),\\
R_{3}& =\frac{\sqrt{N-\cN}}{N-1}\iiiint( 1 \otimes Q(t) w_N Q(t)\otimes Q(t))(x,y;x',y')\times \\
& \qquad \qquad \qquad \qquad \qquad \qquad \qquad  \times \overline{u(t,x)} a^*_y a_{x'} a_{y'} \,\d x \d y \d x' \d y',\\
R_{4}&=R_4^*= \frac{1}{2(N-1)}\iiiint({Q(t)}\otimes{Q(t)}w_N Q(t)\otimes Q(t))(x,y;x',y')\times \\
&  \qquad \qquad \qquad \qquad \qquad \qquad \qquad \qquad \qquad   \times a^*_x a^*_y a_{x'} a_{y'} \,\d x \d y \d x' \d y'.
\end{align*}
Here, in $R_0$ and $R_1$ we write $w_N$ for the function $w_N(x)$, while in $R_3$ and $R_4$ we write $w_N$ for the two-body multiplication operator $w_N(x-y)$.

Bogoliubov's approximation suggests that the error term $\cE_{N}(t)$ should be small. In the following we will justify this on the sectors of few particles. 

\begin{lemma}[Bogoliubov approximation] \label{lem:Bog-app} For every $1\le m \le N$ and $\eps>0$ we have the quadratic form estimate
\begin{align}\label{eq:Bog-Er-1}
\pm \1^{\le m}\cE_{N}(t)\1^{\le m} &\le C_{t,\eps} N^\eps \sqrt{\frac{m}{N}} \dGamma(1-\Delta) \\
\label{eq:Bog-Er-2}
\pm \1^{\le m} \partial_t \cE_{N}(t)\1^{\le m}  &\le C_{t,\eps} N^\eps \sqrt{\frac{m}{N}} \dGamma(1-\Delta)  \\
\label{eq:Bog-Er-3}
\pm  \1^{\le m} i [\cE_{N}(t),\cN]\1^{\le m}  &\le C_{t,\eps} N^\eps \sqrt{\frac{m}{N}} \dGamma(1-\Delta) 
\end{align}
\end{lemma}

\begin{proof} We will use the decomposition \eqref{eq:E-error-R}. We will denote by $\Phi$ an arbitrary normalized vector in $\cF_+^{\le m}$.  Note that \eqref{eq:Bog-Er-1} and \eqref{eq:Bog-Er-2} follow from
\begin{equation}\label{eq:abs-Rj}
|\langle \Phi, R_j \Phi\rangle|\le C_{t,\eps} N^\eps \sqrt{\frac{m}{N}} \langle \Phi, \dGamma(1-\Delta) \Phi\rangle, \quad \forall 0\le j\le 4
\end{equation}
and
\begin{equation}\label{eq:abs-dt-Rj}
|\langle \Phi, \partial_t R_j \Phi\rangle|\le C_{t,\eps} N^\eps \sqrt{\frac{m}{N}} \langle \Phi, \dGamma(1-\Delta) \Phi\rangle,\quad \forall 0\le j\le 4
\end{equation}
respectively. Moreover, since 
$$ [R_0,\cN]=[R_4,\cN]=0,\quad [R_1,\cN]=R_1,\quad [R_2,\cN]=-2R_2, \quad [R_3,\cN]=R_3,$$
the commutator bound \eqref{eq:Bog-Er-3} also follows from \eqref{eq:abs-Rj}. 

Now we prove \eqref{eq:abs-Rj}-\eqref{eq:abs-dt-Rj} term by term. 

\medskip

\noindent
$\boxed{j=0}$ Recall that
$$
R_0 = \d\Gamma\Big( Q(t)[w_N*|u(t)|^2+\widetilde{K}_1(t) -\mu_N(t)]Q(t) \Big)\frac{1-\cN}{N-1}.
$$
From the operator bounds
$$
\| Q(t)[w_N*|u(t)|^2+\widetilde{K}_1(t) -\mu_N(t)]Q(t) \| \le C_t
$$
we have
$$
\pm R_0 \le C_t \frac{\cN (\cN+1)}{N}.
$$
Consequently, 
$$
\pm \1^{\le m}R_0 \1^{\le m} \le \frac{C_t m}{N}\cN.
$$
Similarly,
$$
\pm \partial_t R_0 \le \frac{C_t \cN (\cN+1)}{N}, \quad \pm  \1^{\le m} \partial_t R_0  \1^{\le m} \le\frac{C_t m}{N}\cN.
$$
\medskip

\noindent 
$\boxed{j=1}$ For every $\Phi\in \cF_+^{\le m}(t)$, by the Cauchy-Schwarz inequality we have
\begin{align*}
\left| \langle \Phi,  R_{1}\Phi \rangle \right| &= \frac{2}{N-1} \left| \left\langle \Phi, \cN\sqrt{N-\cN} a\Big( Q(t)[w_N*|u(t)|^2]u(t)  \Big) \Phi \right\rangle \right| \\
& \le \frac{2}{N-1} \| \cN\sqrt{N-\cN} \Phi\| \Big\| a \Big( Q(t)[w_N*|u(t)|^2]u(t) \Big) \Phi \Big\|.
\end{align*}
Using the elementary inequality $a^*(v)a(v)\le \|v\|_{L^2}^2 \cN$ and 
\begin{align*}
\Big\|  Q(t)[w_N*|u(t)|^2]u(t) \Big\|_{L^2}  \le \| [w_N*|u(t)|^2] u(t)\|_{L^2} \le C_t
\end{align*}
we obtain 
\begin{equation*}
\begin{aligned}
\frac{2}{N-1} & \| \cN\sqrt{N-\cN} \Phi\| \Big\| a \Big( Q(t)[w_N*|u(t)|^2]u(t) \Big) \leq 
 C \frac{\delta }{(N-1)^2} \|\sqrt{N} \cN \Phi\|^2 \\
 &  +C_t \delta^{-1}\|\cN^{1/2} \Phi\|^2 \leq C\frac{m \delta}{N}\langle \Phi, \cN \Phi\rangle + C_t \delta^{-1}\langle \Phi, \cN \Phi\rangle
\end{aligned}
\end{equation*}
which choosing $\delta=\sqrt{N/m}$ yields the result
\begin{align}\label{eq:R1-final-0}
\left| \langle \Phi,  R_{1} \Phi \rangle \right| \le C_t \sqrt{\frac{m}{N}} \langle \Phi, \cN \Phi\rangle.
\end{align}
Similarly, using 
\begin{align*}
\Big\|\partial_t \Big( Q(t)[w_N*|u(t)|^2]u(t) \Big)\Big\|_{L^2}  \le C_t
\end{align*}
we get
\begin{align}\label{eq:R1-final-01}
\left| \langle \Phi,  \partial_t R_{1} \Phi \rangle \right| \le   C_t \sqrt{\frac{m}{N}} \langle \Phi, \cN \Phi\rangle.
\end{align}

\noindent
$\boxed{j=2}$ We have 
\begin{align*}
&\langle \Phi, R_2 \Phi \rangle = \iint  K_2(t,x,y) \Big\langle \Phi, a^*_x a^*_y  \Big(\frac{\sqrt{(N-\cN)(N-\cN-1)}}{N-1}-1\Big) \Phi \Big\rangle \d x \d y \\
&= \iint w_N(x-y) u(t,x)u(t,y) \Big\langle \Phi, a^*_x a^*_y  \Big(\frac{\sqrt{(N-\cN)(N-\cN-1)}}{N-1}-1\Big) \Phi \Big\rangle \d x \d y.
\end{align*}
Here we have replaced $K_2(t)=Q(t)\otimes Q(t) \widetilde K_2(t)$ by $\widetilde K_2(t)$, namely ignored the projection   projection $Q(t)$, because $\Phi$ belongs to the excited Fock space $\cF_+(t)$. By the Cauchy-Schwarz inequality, \begin{align*} 
 \left| \langle \Phi, R_2 \Phi \rangle \right| & \le \iint  |w_N(x-y)| . |u(t,x)|. |u(t,y)|.  \| (\cN+1)^{-1/2} a_x a_y \Phi \| \d x \d y  \nn\\
&\qquad \times  \left\| (\cN+1)^{1/2}\Big(\frac{\sqrt{(N-\cN)(N-\cN-1)}}{N-1}-1\Big) \Phi  \right\| \\
&\le \left( \iint  |w_N(x-y)| |u(t,x)|^2 |u(t,y)|^2 \d x \d y \right)^{1/2} \\
&\qquad \times \left( \iint |w_N(x-y)| \| (\cN+1)^{-1/2}a_x a_y \Phi \|^2 \d x \d y \right)^{1/2}\\
&\qquad \times  \left\| (\cN+1)^{1/2}\Big(\frac{\sqrt{(N-\cN)(N-\cN-1)}}{N-1}-1\Big) \Phi  \right\|.
\end{align*}
On the truncated Fock space $\cF_{+}^{\le m}$, we have the operator inequality
\begin{align*}
\left( (\cN+1)^{1/2} \Big(\frac{\sqrt{(N-\cN)(N-\cN-1)}}{N-1}-1\Big) \right)^2  \le \frac{2m}{N} (\cN+1).
\end{align*}
Moreover, from Sobolev's embedding in $\R^2$ we have the operator estimate on $L^2((\R^2)^2)$
\begin{equation}\label{eq:Sobolev_interaction}
|w_N(x-y)| \le C_\eps N^\eps (1-\Delta_x), \quad \forall \eps>0. 
\end{equation}
Consequently, 
\begin{align*} 
&\iint |w_N(x-y)| a_x^* a_y^* (\cN+1)^{-1} a_x a_y \d x \d y\\
& =(\cN+3)^{-1} \iint  a_x^* a_y^* |w_N(x-y)| a_x a_y \d x \d y\\
&\le C_\eps N^\eps (\cN+3)^{-1} \iint   a_x^* a_y^* (1-\Delta)_x a_x a_y \d x \d y\\
&\le C_\eps N^\eps \dGamma(1-\Delta) .
\end{align*}
Thus, since
$$
\left( \iint  |w_N(x-y)| |u(t,x)|^2 |u(t,y)|^2 \d x \d y \right)^{1/2}\leq C_t
$$
we obtain
$$
\langle \Phi,R_2 \Phi\rangle \leq C_t \delta C_\eps N^\eps \langle \Phi,\dGamma(1-\Delta)\Phi \rangle +C_t \delta^{-1}\frac{m}{N} \langle \Phi,(\cN+1)\Phi \rangle
$$
which choosing $\delta=\sqrt{m/N}$ finally leads to
$$
|\langle \Phi, R_2 \Phi \rangle| \le C_\eps N^\eps \sqrt{\frac{m}{N}}\langle \Phi, \dGamma(1-\Delta) \Phi \rangle. 
$$

Next, we consider 
\begin{align*}
&\langle \Phi, \partial_t R_2  \Phi \rangle =\iint \Big[ (\partial_t Q(t) \otimes 1 + 1\otimes \partial_t Q(t)) \widetilde K_2(t,x,y) + \partial_t \widetilde K_2(t,x,y)\Big] \times \\
&\qquad \qquad \qquad \times   \Big \langle \Phi, a^*_x a^*_y \Big(\frac{\sqrt{(N-\cN)(N-\cN-1)}}{N-1}-1\Big) \Phi \Big\rangle \d  x \d y.
\end{align*}
Here we have used the decomposition
$$
\partial_t K_2(t)= \partial_t Q(t) \otimes Q(t) \widetilde K_2(t) + Q(t) \otimes \partial_t Q(t)  \widetilde K_2(t) + Q(t)\otimes Q(t) \partial_t \widetilde K_2(t)
$$
and omitted the projection $Q(t)$ again using $\Phi \in \cF^{\le m}_{+}(t)$. The term involving $\partial_t \widetilde K_2(t,x,y)$ can be bounded as above. The term involving $(\partial_t Q(t) \otimes 1)  \widetilde K_2(t,x,y)$ is bounded as
\begin{align*}
&\left| \iint \big(\partial_t Q(t) \otimes 1\big) \widetilde K_2(t,x,y)   \Big \langle \Phi, a^*_x a^*_y \Big(\frac{\sqrt{(N-\cN)(N-\cN-1)}}{N-1}-1\Big) \Phi \Big\rangle \d  x \d y \right| \\
&\le \left( \iint \Big| \big(\partial_t Q(t) \otimes 1\big) \widetilde K_2(t,x,y) \Big|^2 \d x \d y \right)^{1/2}\left( \iint \| (\cN+1)^{-1/2}a_x a_y \Phi\|^2 \d x \d y \right)^{1/2} \\
&\quad \times \left\| (\cN+1)^{1/2}\left(\frac{\sqrt{(N-\cN)(N-\cN-1)}}{N-1}-1\right) \Phi  \right\| \nn \\
&\le \frac{C_t m}{N}\langle \Phi, (\cN+1) \Phi \rangle.
\end{align*}
Here we used \eqref{eq:partialtK2bound}. The term involving $1 \otimes \partial_t Q(t)$ can be bounded similarly. Thus
$$
\left| \langle \Phi, \partial_t R_2  \Phi \rangle  \right| \le \frac{C_{t,\eps} N^\eps m}{N}\langle \Phi, \dGamma(1-\Delta) \Phi \rangle.
$$
Since $m\leq N$ the desired estimate follows.
\medskip

\medskip

\noindent
$\boxed{j=3}$ By the Cauchy-Schwarz inequality and Sobolev's inequality \eqref{eq:Sobolev_interaction} we have 
\begin{align} \label{eq:R3-final-0}  
\left| \langle \Phi, R_3 \Phi \rangle \right| &=  \frac{1}{N-1}  \left| \iint  w_N(x-y) \overline{u(t,x)} \Big \langle \Phi,\sqrt{N-\cN} a^*_y a_{y} a_x \Phi \Big\rangle \,\d x \d y \right| \nn\\
& \le \frac{1}{N-1} \iint  |w_N(x-y)| |u(t,x)| \cdot \| a_y \sqrt{N-\cN} \Phi \| \cdot \| a_{y}a_x  \Phi \| \d x \d y \nn\\
& \le \frac{\|u(t,\cdot)\|_{L^\infty}}{N-1} \left( \iint |w_N(x-y)| \|a_x a_y \Phi\|^2 \d x \d y\right)^{1/2} \nn \\
&\qquad\qquad\qquad \times \left( \iint |w_N(x-y)| \| a_y \sqrt{N-\cN} \Phi \|^2 \d x \d y \right)^{1/2} \nn\\
&\le  \frac{C_t}{N}  \langle \Phi,C_\eps N^\eps \dGamma(1-\Delta)\cN \Phi \rangle^{1/2} \langle \Phi, \cN N \Phi \rangle^{1/2}\nn\\
&\le \frac{C_{t,\eps} N^\eps m^{1/2}}{N^{1/2}} \langle \Phi, \dGamma(1-\Delta)  \Phi \rangle.
\end{align}
Next, 
\begin{align*}
\langle \Phi, \partial_t R_3 \Phi\rangle & = \frac{1}{N-1}  \iiiint \Big[   \big( 1 \otimes Q(t) w_N Q(t)\otimes Q(t)\big)(x,y;x',y') \overline{\partial_t u(t,x)} \\
&\qquad \qquad  +  \Big( \partial_t \big( 1 \otimes Q(t) w_N Q(t)\otimes Q(t)\big)\Big)(x,y;x',y') \overline{ u(t,x)} \Big] \times \\
& \qquad \qquad \qquad \qquad\qquad  \times \langle \Phi, \sqrt{N-\cN} a^*_y a_{x'} a_{y'}  \Phi \rangle \,\d x \d y \d x' \d y'.
\end{align*}
The term involving $\partial_t u(t,x)$ can  be estimated similarly to  \eqref{eq:R3-final-0}. For the other term, we decompose
\begin{align*}\partial_t \big( 1 \otimes Q(t) w_N Q(t)\otimes Q(t)\big)&= 1 \otimes \partial_t Q(t) w_N Q(t)\otimes Q(t)+1 \otimes Q(t) w_N \partial_t Q(t)\otimes Q(t)\\
&\quad +1 \otimes Q(t) w_N Q(t)\otimes \partial_t Q(t)
\end{align*}
We will use the kernel estimate
\bq \label{eq:kernel-dtQ}
|(\partial_t Q(t))(z;z')|= |\partial_t u(t,z) \overline{u(t,z')}+ u(t,z) \overline{\partial_t u(t,z')}| \le q(z) q(z')
\eq
where 
$$
q(t,z):=|u(t,z)| + |\partial_t u(t,z)| , \quad \|q(t, \cdot)\|_{L^2} \le C_t, \quad \|q(t,\cdot)\|_{L^\infty}  \le C_t.
$$
For the first term involving $1 \otimes \partial_t Q(t) w_N Q(t)\otimes Q(t)$, we can estimate 
\begin{align*}
&\frac{1}{N-1} \left| \iiiint  (1 \otimes \partial_t Q(t) w_N Q(t)\otimes Q(t)) (x,y;x',y')  \overline{u(t,x)}  \times \right.\\
& \qquad \qquad \qquad \qquad \qquad \qquad\left. \times \Big\langle \Phi, \sqrt{N-\cN} a^*_y a_{x'} a_{y'}  \Phi \Big\rangle \,\d x \d y \d x' \d y' \right|\\
&=\frac{1}{N-1}  \left|  \iiiint  (\partial_t Q(t))(y;y') w_N(x-y') \delta(x-x')  \overline{ u(t,x)} \times \right.\\
&\qquad \qquad \qquad \qquad \qquad \qquad  \left.\times \Big \langle \Phi,\sqrt{N-\cN} a^*_y a_x a_{y'}  \Phi \Big\rangle \,\d x \d y \d x' \d y' \right| \nn \\
&\le \frac{1}{N-1}  \iiint    |q(t,y)|\, |q(t,y')|\, |w_N(x-y')|  |q(t,x)| \times \\
&\qquad \qquad \qquad \qquad \qquad \qquad  \times \| a_y \sqrt{N-\cN} \Phi\| \|a_x a_{y'}  \Phi \|\,\d x \d y \d y' \nn \\
&\le \frac{\|q(t,\cdot)\|_{L^\infty}}{N-1} \left( \int |q(t,y)|^2 \d y\right)^{1/2} \left( \int \| a_y \sqrt{N-\cN} \Phi \|^2 \d y\right)^{1/2} \times \\
& \times \left( \iint w_N(x-y') |q(t,x)|^2 \d x \d y'\right)^{1/2} \left( \iint w_N(x-y') \| a_x a_{y'} \Phi\|^2 \d x \d y'\right)^{1/2} \\
&\le \frac{C_t}{N} \langle \Phi, \cN N \Phi\rangle^{1/2}  \langle \Phi, C_\eps N^\eps \dGamma(1-\Delta)\cN \Phi\rangle^{1/2}\\
&\le \frac{C_{t,\eps} N^\eps m^{1/2}}{N^{1/2}} \langle \Phi, \dGamma(1-\Delta) \Phi \rangle .
\end{align*}
For the term involving $1 \otimes  Q(t) w_N  \partial_t Q(t) \otimes Q(t)$, we have 
\begin{align*}
&\frac{1}{N-1} \left| \iiiint  (1 \otimes  Q(t) w_N  \partial_t Q(t) \otimes Q(t))(x,y;x',y')   \overline{u(t,x)}  \times \right. \\
& \qquad \qquad \qquad \qquad \qquad \qquad \left. \times \Big \langle \Phi, \sqrt{N-\cN} a^*_y a_{x'} a_{y'}  \Phi \Big\rangle \,\d x \d y \d x' \d y' \right|\\
&=  \frac{1}{N-1}  \left| \iiiint   w_N(x-y) (\partial_tQ(t))(x,x')  \delta(y-y') \overline{ u(t,x)} \times \right.\\
&\qquad \qquad \qquad \qquad \qquad \qquad \left.\times \Big \langle \Phi,\sqrt{N-\cN} a^*_y a_{x'} a_{y'}  \Phi \Big\rangle \,\d x \d y \d x' \d y' \right| \nn \\
&\le \frac{1}{N-1}   \iiint   |w_N(x-y)|\, |q(t,x)|\, |q(t,x')|   \times \\
&\qquad \qquad \qquad \qquad\qquad \qquad  \times \| a_y \sqrt{N-\cN} \Phi\| \|a_{x'} a_{y}  \Phi\| \,\d x \d y \d x' \nn \\
&\le \frac{ \| q(t,\cdot)\|_{L^\infty} ^2}{N-1} \|w_N\|_{L^1} \left( \iint \|a_{x'} a_{y}  \Phi\|^2 \,\d x' \d y \right)^{1/2} \times  \\
&\quad \times   \left( \iint |q(t,x')|^2 \| a_y \sqrt{N-\cN} \Phi\|^2 \d x' \d y \right)^{1/2} \\
&\le \frac{C}{N}  \langle \Phi, \cN^2 \Phi\rangle^{1/2} \langle \Phi, \cN N \Phi\rangle^{1/2}  \le \frac{Cm^{1/2}}{N^{1/2}} \langle \Phi, \cN  \Phi\rangle.
\end{align*}
The term involving $1 \otimes  Q(t) w_N   Q(t) \otimes \partial_t Q(t)$ is bounded similarly as above. Thus 
\begin{align*}
\left| \langle \Phi, \partial_t R_3 \Phi \rangle\right| \le \frac{C_{t,\eps} N^\eps m^{1/2}}{N^{1/2}} \langle \Phi, \dGamma(1-\Delta)  \Phi \rangle.
\end{align*}

\medskip

\noindent
$\boxed{j=4}$ By Sobolev's inequality \eqref{eq:Sobolev_interaction}
$$
|R_4| \le \frac{C_\eps N^\eps}{N} \dGamma(1-\Delta)\cN.
$$
Therefore,
$$
\pm \1^{\le m} R_4 \1^{\le m} \le  \frac{C_\eps N^\eps m}{N} \dGamma(1-\Delta).
$$
Next, we consider  
\begin{align*}
&\langle \Phi, \partial_t R_4\Phi\rangle = \frac{1}{2(N-1)} \Re \iiiint \partial_t \Big ( Q(t) \otimes Q(t) w_N Q(t)\otimes Q(t)\Big)(x,y;x',y') \\
& \qquad \qquad \qquad \qquad \qquad \qquad\qquad \qquad\times \langle \Phi, a_x^* a^*_y a_{x'} a_{y'}  \Phi \rangle \,\d x \d y \d x' \d y'.
\end{align*}
Let us decompose $\partial_t \big ( Q(t) \otimes Q(t) w_N Q(t)\otimes Q(t)\big)$ into four terms, and consider for example $\partial_t Q(t) \otimes Q(t) w_N Q(t)\otimes Q(t)$. Using \eqref{eq:kernel-dtQ} again and Sobolev's inequality \eqref{eq:Sobolev_interaction},  we have
\begin{align*}
&\frac{1}{N-1} \left|  \iiiint \Big(\partial_t Q(t) \otimes Q(t) w_N Q(t)\otimes Q(t)\Big)(x,y;x',y')\times \right.\\
& \qquad \qquad \qquad \qquad \qquad \qquad \qquad \qquad \left. \times \Big\langle \Phi, a_x^* a^*_y a_{x'} a_{y'}  \Phi \Big\rangle \,\d x \d y \d x' \d y' \right|\\
&=  \frac{1}{N-1} \left| \iiiint (\partial_t Q(t))(x,x') w_N(x'-y) \delta(y-y')\times \right. \\
& \qquad \qquad \qquad \qquad \qquad \qquad\qquad \qquad \left. \times \Big\langle \Phi, a_x^* a^*_y a_{x'} a_{y'}  \Phi \Big\rangle \,\d x \d y \d x' \d y' \right|\\
&\le \frac{1}{N-1} \iiint |q(t,x)|\, |q(t,x')| \, |w_N(x'-y)|\, \| a_x a_y \Phi\| \|a_{x'} a_{y} \Phi\| \,\d x \d y \d x'\\
&\le \frac{\|q(t,\cdot)\|_{L^\infty}}{N-1} \left( \iiint |w_N(x'-y)| \| a_x a_y \Phi\|^2 \d x \d y \d x'  \right)^{1/2} \times \\
&\qquad \qquad \qquad \qquad  \times \left( \iiint |w_N(x'-y)|\,|q(t,x)|^2 \| a_{x'} a_y \Phi\|^2 \d x \d y \d x'  \right)^{1/2} \\
&\le \frac{C_t}{N}   \langle \Phi, \cN^2 \Phi\rangle^{1/2} C_\eps N^\eps \langle \Phi, \dGamma(1-\Delta) \cN \Phi\rangle^{1/2} \\
&\le \frac{C_{t,\eps} N^\eps m}{N}  \langle \Phi, \dGamma(1-\Delta) \Phi\rangle.
\end{align*}
Thus
\begin{align*}
\left| \langle \Phi, \partial_t R_4\Phi \rangle \right| \le \frac{C_{t,\eps} N^\eps m}{N}  \langle \Phi, \dGamma(1-\Delta) \Phi\rangle.\end{align*}
This completes the proof.  
\end{proof}

As a simple consequence of the above estimates, we have an a-priori upper bound for the kinetic energy of $\Phi_N$. 

\begin{lemma}\label{lem:kinetic-PhiN} For all $\beta>0$ we have
$$
\langle \Phi_N(t), \dGamma(1-\Delta) \Phi_N(t) \rangle \le C_{t,\eps} (N+N^{2\beta}). 
$$
\end{lemma}
\begin{proof} Since $\Psi_N(t)$ is the Schr\"odinger dynamics, it preserves energy:  
$$\langle \Psi_N(t), H_N \Psi_N(t)\rangle= \langle \Psi_N(0), H_N \Psi_N(0)\rangle.$$ 
From \cite[Lemma 2]{LewNamRou-17} we have the lower bound 
$$
H_N -\eps \sum_{i=1}^N (-\Delta_i) \ge -C_\eps N^{2\beta}. 
$$
with fixed $\eps>0$ small enough. We conclude that
\begin{equation*}
\begin{aligned}
\langle \Psi_N, \sum_{i=1}^N (1-\Delta_i)  \Psi_N \rangle \le N+C_\eps N^{2\beta}+C_\eps\langle \Psi_N, H_N \Psi_N \rangle \leq C_\eps (N+N^{2\beta})
\end{aligned}
\end{equation*}
which follows from 
$$
\langle \Psi_N(0), H_N \Psi_N(0)\rangle\leq CN.
$$
The last inequality follows from the assumption on $\Psi_N(0)$ and bounds derived in Lemmas \ref{lem:bH-kinetic} and \ref{lem:Bog-app}. Finally, we have
\begin{equation*}
\begin{aligned}
\langle \Psi_N & (t), \sum_{i=1}^N (1-\Delta_i)  \Psi_N(t) \rangle=\langle \Phi_N(t), U_N(t) \dGamma (1-\Delta)U_N^* (t)   \Phi_N(t) \rangle \\
&  = \langle \Phi_N(t),\dGamma (1-\Delta)\Phi_N (t) \rangle +\int |\nabla u(t)|^2 dx   \langle \Phi_N(t), (N-\cN)\Phi_N(t)\rangle \\
&\qquad+ \left[\langle \Phi_N(t), \sqrt{N-\cN} a(Q(t) (-\Delta) u(t) ) \Phi_N(t)\rangle +\text{h.c.}\right].
\end{aligned}
\end{equation*}
Since 
$$\int |\nabla u(t)|^2 dx   \langle \Phi_N(t), (N-\cN)\Phi_N(t)\rangle \leq C_t N$$
and 
\begin{equation*}
\begin{aligned}
\Big|\langle \Phi_N(t),& \sqrt{N-\cN} a(Q(t) (-\Delta) u(t) ) \Phi_N(t)\rangle\Big|\leq \|\sqrt{N-\cN} \Phi_N \|\, \|a(Q(t) (-\Delta) u(t) ) \Phi_N(t)\| \\
& \leq N+ \|Q(t) (-\Delta) u(t) \|_2^2 N\leq C_t N
\end{aligned}
\end{equation*}
we conclude that
$$\langle \Phi_N(t),\dGamma (1-\Delta)\Phi_N (t) \rangle \leq C_{t,\eps}(N+N^{2\beta}).$$
\end{proof}

\section{Truncated dynamics} \label{sec:loc}

As we explained, the main goal is to compare $\Phi_N$ and $\Phi$. Instead of doing this directly, we introduce an intermediate dynamics living in the sector of very few particles. Related ideas have been used in our works \cite{NamNap-17,BNNS} on defocusing 3D systems. Take 
$$M = N^{1-\delta}, \quad \delta\in (0,1)$$
and let $\Phi_{N,M}(t)$ be the solution to
\begin{equation} \label{eq:def-PhiNM}
i\partial_t \Phi_{N,M}(t) = \1^{\le M}\cG_{N}(t)\1^{\le M}\Phi_{N,M}(t), \quad \Phi_{N,M}(0)= \1^{\le M} \Phi_{N}(0).
\end{equation}
Putting differently, $\Phi_{N,M}(t)$ is the dynamics in the truncated excited Fock space $\cF_+^{\le M}$, generated by the quadratic form $\cG_{N}(t)$ restricted on $\cF_+^{\le M}$. 

The existence and uniqueness of $\Phi_{N,M}(t)$ follows from \cite[Theorem 7]{LewNamSch-15}. Let us briefly explain why $\Phi_{N,M}(t)$ indeed belongs to $\cF_+^{\le M}$. The identity 
$$\frac{d}{dt}\|\1^{>M}\Phi_{N,M}(t)\|^2 =\langle  \Phi_{N,M}(t), i[\1^{\le M}\cG_{N}(t)\1^{\le M},\1^{>M}]\Phi_{N,M}(t)\rangle=0$$
together with the initial condition $\1^{>M}\Phi_{N,M}(0)=0$ imply that $\1^{>M}\Phi_{N,M}(t)=0$ and thus indeed $\Phi_{N,M}(t)$ belongs to the truncated space with no more than $M$ particles. 

The fact that $\Phi_{N,M}(t)$ belongs to the excited space follows from the same argument applied to  $\|a(u(t))\Phi_{N,M}(t)\|^2$ (see also proof of \cite[Thm. 7]{LewNamSch-15}. 

 The kinetic energy of this truncated dynamics is the subject of the following

\begin{lemma} \label{lem:kinetic-PhiNM} When $M = N^{1-\delta}$ with a constant $\delta\in (0,1)$ we have
$$ \langle \Phi_{N,M}(t), \dGamma(1-\Delta) \Phi_{N,M}(t) \rangle \le C_{t,\eps} N^\eps,\quad \forall \eps>0.$$
\end{lemma}

\begin{proof} By using Lemmas \ref{lem:bHt-dbHt}, \ref{lem:Bog-app} and the assumption $M\le N^{1-\delta}$ we have the following quadratic form estimate on $\cF^{\le M}:$
\begin{align*}
\pm (\cG_{N}(t)+\dGamma(\Delta)) &\le \frac{1}{2} \dGamma(1-\Delta) + C_{t,\eps} (\cN + N^\eps), \quad \forall \eps>0, \\
\pm \partial_t \cG_{N}(t) &\le C_t \dGamma(1-\Delta),\\
\pm i[\cG_{N}(t),\cN] &\le C_t \dGamma(1-\Delta).
\end{align*}
The desired estimate follows by applying Gronwall's argument as in the proof of Lemma \ref{lem:bH-kinetic}, with $A(t)$ replaced by $\1^{\le M}\cG_{N}(t)\1^{\le M} + C_{t,\eps} (\cN + N^\eps)$. 
\end{proof}

The main result of this section is the following comparison. 

\begin{lemma} \label{lem:PhiN-PhiNM}  When $M = N^{1-\delta}$ with a constant $\delta\in (0,1)$,
 $$\|\Phi_{N}(t)-\Phi_{N,M}(t)\|^2\le C_{t,\eps} N^\eps \Big( \frac{ N^\beta}{M}  + \frac{1}{M^{1/2}}\Big), \quad \forall \eps>0.$$
\end{lemma}

\begin{proof} Note that 
$$
\|\Phi_N(t)\|=\|\Phi_N(0)\|= \|\1^{\le N}\Phi(0)\|\le 1, \quad \|\Phi_{N,M}(t)\|=\|\Phi_{N,M}(0)\|=\|\1^{\le M} \Phi(0)\|\le 1.
$$
Therefore,
$$\|\Phi_{N}(t)-\Phi_{N,M}(t)\|^2\le 2\Big(1-\Re \langle \Phi_{N}(t),\Phi_{N,M}(t)\rangle \Big).$$
Take a parameter $M/2\le m\le M-3$ and write 
$$
\langle \Phi_{N}(t),\Phi_{N,M}(t)\rangle =\langle \Phi_{N}(t),\1^{\le m}\Phi_{N,M}(t)\rangle +\langle \Phi_{N}(t),\1^{>m}\Phi_{N,M}(t)\rangle.
$$
For the many-particle sectors, using the Cauchy-Schwarz inequality and Lemma \ref{lem:kinetic-PhiNM} we can estimate 
\begin{align} \label{eq:PhiN-PhiNM->m}
|\langle \Phi_{N}(t),\1^{>m}\Phi_{N,M}(t)\rangle| &\le \|\Phi_{N}(t)\|. \|\1^{>m}\Phi_{N,M}(t)\| \nn\\
&\le \Big\langle \Phi_{N,M}(t), (\cN/m) \Phi_{N,M}(t) \Big\rangle^{1/2} \nn\\
&\le C_{t,\eps} N^\eps M^{-1/2}. 
\end{align}
For the few-particle sectors, we use the equations of $\Phi_N(t)$ and $\Phi_{N,M}(t)$. We have
\begin{align*}
\frac{d}{dt} \langle \Phi_{N}(t),\1^{\le m}\Phi_{N,M}(t)\rangle &= i \Big\langle  \Phi_{N}(t), \Big(\cG_{N}(t)\1^{\le m} - \1^{\le m} \1^{\le M} \cG_{N}(t) \1^{\le M} \Big) \Phi_{N,M}(t)  \Big\rangle\\
&= \Big\langle  \Phi_{N}(t), i[\cG_{N}(t),\1^{\le m}] \Phi_{N,M}(t)  \Big\rangle
\end{align*}
Here in the last equality we have used 
\begin{equation} \label{eq:1<M1<m}
\1^{\le m} \1^{\le M} \cG_{N}(t)\1^{\le M}=\1^{\le m} \cG_{N}(t)
\end{equation}
which follows from the choice $m\le M-3$ and the fact that $\cG_N(t)$ contains at most 2 creation operators and at most 2 annihilation operators. 

Now we average over $m\in [M/2,M-3]$. 
\begin{lemma} \label{lem:com-G-1} When $M = N^{1-\delta}$ with a constant $\delta\in (0,1)$, then on the excited Fock space we have
$$
\pm \frac{1}{M/2-2}\sum_{m=M/2}^{M-3} i[\cG_{N}(t),\1^{\le m}]  \le \frac{C_{t,\eps} N^\eps}{M}   \dGamma(1-\Delta).$$
\end{lemma}

\begin{proof} We decompose 
$$[\cG_{N}(t),\1^{\le m}] = \1^{\le m} \cG_{N}(t)\1^{> m} - \1^{> m}\cG_{N}(t)\1^{\le m}.$$
We will focus on $\1^{> m}\cG_{N}(t)\1^{\le m}.$ The other term can be treated similarly.
Let us denote
\begin{equation*}
\begin{aligned}
A_1 & =\frac12  \iiiint( Q(t)\otimes Q(t)w_N Q(t)\otimes 1)(x,y;x',y')u(t,x') a^*_{x} a^*_{y} a_{y'} \,\d x \d y \d x' \d y',\\
 & \qquad   -a^*(Q(t)[w_N *|u(t)|^2]u(t))\cN =: A_1^3+A_1^1, \\
A_2 & = \frac12 \iint  K_2(t,x,y) a^*_x a^*_y \d x \d y.
\end{aligned}
\end{equation*}
 We then have
\begin{equation*}
\begin{aligned}
\1^{> m}&\cG_{N}(t)\1^{\le m}=\1^{\le N} \1^{> m}\Big( A_1\frac{\sqrt{N-\cN}}{N-1}+A_2 \frac{\sqrt{(N-\cN)(N-\cN-1)}}{N-1} \Big)\1^{\le m}\1^{\le N} \\
& =   A_1\frac{\sqrt{N-\cN}}{N-1}\1(\cN=m)+  A_2 \frac{\sqrt{(N-\cN)(N-\cN-1)}}{N-1} \1(m-1\leq \cN\leq m).  
\end{aligned}
\end{equation*}
Here we used the fact that $A_1$ contains one creation operator, while $A_2$ contains two creation operators. It follows that
\begin{equation*}
\begin{aligned}
&\sum_{m=M/2}^{M-3} \1^{> m}\cG_{N}(t)\1^{\le m}=  A_1\frac{\sqrt{N-\cN}}{N-1}\1(M/2\leq\cN \leq M-3)  \\
&  + A_2 \frac{\sqrt{(N-\cN)(N-\cN-1)}}{N-1}\left[\1(M/2-1 <\cN\leq M-3)+\1(M/2-1 \leq \cN< M-3)\right].
\end{aligned}
\end{equation*}
Let us now bound the first term. For any $X\in \cF_+$, using Cauchy-Schwarz we have 
\begin{equation*}
\begin{aligned}
|\langle & X, A_1^3 \frac{\sqrt{N-\cN}}{N-1}\1(M/2\leq\cN \leq M-3)X\rangle| \leq \frac12 \int |w_N(x-y)| |u(t,x)| \|a_x a_y X\| \times \\
& \times \|a_y \frac{\sqrt{N-\cN}}{N-1}\1(M/2\leq\cN \leq M-3)X\|dxdy \leq \frac{C}{N}\int |w_N(x-y)| \|a_x a_y X\|^2 dxdy \\
& \qquad +  \frac{C}{N}\int |w_N(x-y)||u(t,x)|^2  \|a_y \sqrt{N-\cN}\1(M/2\leq\cN \leq M-3)X\|^2 dxdy \\
& \leq \frac{C_\eps N^\eps}{N} \langle X, \dGamma(1-\Delta)\cN X \rangle +  C_t \langle X, \cN X \rangle \leq C_{t,\eps}N^{\eps} \dGamma(1-\Delta)
\end{aligned}
\end{equation*}
where we used \eqref{eq:Sobolev_interaction} and $\|u(t)\|_{\infty} \leq C_t$. The terms involving $A_1^1$ and $A_2$ can be bounded in the same way. This ends the proof.
\end{proof}

\begin{remark} From the proof, we also obtain 
$$
\pm \frac{1}{M/2-2}\sum_{m=M/2}^{M-3} i[\bH(t),\1^{\le m}]  \le \frac{C_{t,\eps} N^\eps}{M}   \dGamma(1-\Delta).$$
\end{remark}

Now we come back to the proof of Lemma \ref{lem:PhiN-PhiNM}. Using Lemma \ref{lem:com-G-1}, the Cauchy-Schwarz inequality and the kinetic estimates in Lemmas \ref{lem:kinetic-PhiN}, \ref{lem:kinetic-PhiNM}, we can estimate
\begin{align*}
&\Big| \frac{1}{M/2-2}\sum_{m=M/2}^{M-3}  \frac{d}{dt} \langle \Phi_{N}(t),\1^{\le m}\Phi_{N,M}(t)\rangle \Big| \\
&= \Big| \frac{1}{M/2-2}\sum_{m=M/2}^{M-3}  \Big\langle  \Phi_{N}(t), i[\cG_{N}(t),\1^{\le m}] \Phi_{N,M}(t)  \Big\rangle \Big| \\
&\le \frac{C_{t,\eps} N^{\eps}}{M} \Big\langle  \Phi_{N}(t), \dGamma(1-\Delta)\Phi_{N}(t)  \Big\rangle^{1/2} \Big\langle  \Phi_{N,M}(t), \dGamma(1-\Delta)\Phi_{N,M}(t)  \Big\rangle^{1/2}\\
&\le \frac{C_{t,\eps} N^{\frac32 \eps} N^\beta}{M}, \quad \forall \eps>0.
\end{align*}
Consequently,
\begin{align*}
&\Re \frac{1}{M/2-2}\sum_{m=M/2}^{M-3}   \langle \Phi_{N}(t),\1^{\le m}\Phi_{N,M}(t)\rangle \\
&\ge \Re \frac{1}{M/2-2}\sum_{m=M/2}^{M-3}   \langle \Phi_{N}(0),\1^{\le m}\Phi_{N,M}(0)\rangle - \frac{C_{t,\eps} N{\frac32 \eps} N^\beta}{M}.
\end{align*}
Moreover, since $\Phi_N(0)=\1^{\le N}\Phi(0)$, $\Phi_{N,M}(0)=\1^{\le M}\Phi(0)$ and $\langle \Phi(0),\cN \Phi(0)\rangle \le C$ we have
\begin{align*}
\langle \Phi_{N}(0),\1^{\le m}\Phi_{N,M}(0)\rangle &= \langle \Phi(0),\1^{\le m}\Phi(0)\rangle \\
&=1 -  \langle \Phi(0),\1^{>m}\Phi(0)\rangle \\
&\ge 1 -  \langle \Phi(0),(\cN/m)\Phi(0)\rangle \ge 1 -\frac{C}{M}.
\end{align*}
Thus
\begin{align*}
\Re \frac{1}{M/2-2}\sum_{m=M/2}^{M-3}   \langle \Phi_{N}(t),\1^{\le m}\Phi_{N,M}(t)\rangle \ge 1- \frac{C}{M}- \frac{C_{t,\eps} N^\eps N^\beta}{M}.
\end{align*}
Recall that from \eqref{eq:PhiN-PhiNM->m} we have immediately
\begin{align*}
\pm \Re \frac{1}{M/2-2}\sum_{m=M/2}^{M-3}   \langle \Phi_{N}(t),\1^{> m}\Phi_{N,M}(t)\rangle \le C_{t,\eps} N^\eps M^{-1/2}.
\end{align*}
Combining the latter two estimates, we arrive at
\begin{align*}
\Re \langle  \Phi_{N}(t),  \Phi_{N,M}(t)\rangle \ge 1 - C_{t,\eps }N^\eps \Big( \frac{ N^\beta}{M}  + \frac{1}{M^{1/2}}\Big).
\end{align*}
Therefore, we conclude that
$$
\| \Phi_{N}(t)-  \Phi_{N,M}(t)\|^2 \le 2 (1-\Re \langle  \Phi_{N}(t),  \Phi_{N,M}(t)\rangle) \le C_{t,\eps} N^\eps \Big( \frac{ N^\beta}{M}  + \frac{1}{M^{1/2}}\Big), \quad \forall \eps>0. 
$$
\end{proof}

\section{Proof of Theorem \ref{thm:main} ($d=2$)}\label{sec:main-proof}

In this section, we prove Theorem \ref{thm:main} in case $d=2$. We have compare $\Phi_N$ and $\Phi_{N,M}$ in the previous section. The last main step is the following comparison 

\begin{lemma} \label{lem:com-PhiNM-Phi}When $M = N^{1-\delta}$ with a constant $\delta\in (0,1)$, 
$$\|\Phi_{N,M}(t)-\Phi(t)\|^2 \le C_{t,\eps} N^\eps \Big( \sqrt{\frac{M}{N}} + \frac{1}{M} \Big), \quad \forall \eps>0.
$$ 
\end{lemma}

\begin{proof} We follow the same strategy as in the proof of Lemma \ref{lem:PhiN-PhiNM}. Again, we use
$$
\|\Phi_{N,M}(t)-\Phi(t)\|^2 \le 2 (1-\Re \langle \Phi_{N,M}(t), \Phi(t) \rangle)
$$
and write, with $M/2\le m\le M-3$, 
$$
\langle \Phi_{N,M}(t), \Phi(t) \rangle= \langle \Phi_{N,M}(t), \1^{\le m}\Phi(t) \rangle + \langle \Phi_{N,M}(t),\1^{>m}\Phi(t) \rangle.
$$
For the many-particle sectors, by using the bounds in Lemmas \ref{lem:bH-kinetic}, \ref{lem:kinetic-PhiNM}  we have
\begin{align} \label{eq:many-pa-PhiNM-Phi}
| \langle \Phi_{N,M}(t),\1^{>m}\Phi(t) \rangle| &\le \| \1^{>m} \Phi_{N,M}(t)\|. \|\1^{>m} \Phi(t)\| \nn\\
&\le \langle \Phi_{N,M}(t),(\cN/m)\Phi_{N,M}(t) \rangle^{1/2}\langle \Phi(t),(\cN/m)\Phi(t) \rangle^{1/2}\nn\\
&\le \frac{C_{t,\eps} N^\eps}{M}.
\end{align}
For the few-particle sectors, by using the equations of $\Phi_{N,M}(t)$ and $\Phi(t)$ and the identity (cf. \eqref{eq:1<M1<m})
$$
 \1^{\le M} \cG_{N}(t)\1^{\le M} \1^{\le m}= \cG_{N}(t)\1^{\le m}
$$
we can write 
\begin{align*}
&\frac{d}{dt} \langle \Phi_{N,M}(t),\1^{\le m}\Phi(t)\rangle \\
&= i \Big\langle  \Phi_{N,M}(t), \Big( \1^{\le M} \cG_{N}(t) \1^{\le M} \1^{\le m}-\1^{\le m}\bH \Big) \Phi(t)  \Big\rangle\\
&= i \Big\langle  \Phi_{N,M}(t), \Big( (\cG_{N}(t)-\bH)\1^{\le m}+ [\bH,\1^{\le m}])\Big) \Phi(t)  \Big\rangle.
\end{align*}
For the term involving $(\cG_{N}(t)-\bH)$, by Lemma \ref{lem:Bog-app} we have 
$$
\pm \1^{\le m+2} (\cG_{N}(t)-\bH)\1^{\le m+2}\le  C_{t,\eps} N^\eps \sqrt{\frac{m}{N}} \dGamma(1-\Delta), \quad \forall \eps>0.
$$
Therefore, by the Cauchy-Schwarz inequality and Lemmas \ref{lem:bH-kinetic}, \ref{lem:kinetic-PhiNM}, we can estimate
\begin{align*}
&\Big|  \Big\langle  \Phi_{N,M}(t),  (\cG_{N}(t)-\bH)\1^{\le m}\Phi(t)\Big\rangle \Big| \\
&= \Big|  \Big\langle  \Phi_{N,M}(t), \1^{\le m+2} (\cG_{N}(t)-\bH) \1^{\le m+2} \1^{\le m}\Phi(t)\Big\rangle \Big| \\
&\le C_{t,\eps} N^{\eps} \sqrt{\frac{m}{N}} \Big\langle \Phi_{N,M}(t), \dGamma(1-\Delta) \Phi_{N,M}(t) \Big\rangle^{1/2} \Big\langle \1^{\le m} \Phi(t), \dGamma(1-\Delta) \1^{\le m} \Phi(t) \Big\rangle^{1/2}\\
&\le C_{t,\eps} N^{2\eps} \sqrt{\frac{M}{N}}, \quad \forall \eps>0.
\end{align*}

For the part involving $i[\bH,\1^{\le m}]$, we will take the average over $m\in [M/2,M-3]$. Recall the remark after Lemma \ref{lem:com-G-1}:
$$
\pm \frac{1}{M/2-2}\sum_{m=M/2}^{M-3} i[\bH(t),\1^{\le m}]  \le \frac{C_{t,\eps} N^\eps}{M}   \dGamma(1-\Delta).$$
By the Cauchy-Schwarz inequality and Lemmas \ref{lem:bH-kinetic}, \ref{lem:kinetic-PhiNM} again, we can thus bound
\begin{align*}
&\Big| \frac{1}{M/2-2}\sum_{m=M/2}^{M-3}  \Big\langle  \Phi_{N,M}(t),  i[\bH,\1^{\le m}]\Phi(t)\Big\rangle \Big| \\
&\le \frac{C_{t,\eps} N^{\eps}}{M}  \Big\langle \Phi_{N,M}(t), \dGamma(1-\Delta) \Phi_{N,M}(t) \Big\rangle^{1/2} \Big\langle \Phi(t), \dGamma(1-\Delta) \Phi(t) \Big\rangle^{1/2}\\
&\le \frac{C_{t,\eps} N^{2\eps}}{M}, \quad \forall \eps>0.
\end{align*}
Thus we have proved that
\begin{align*}
&\Big| \frac{1}{M/2-2}\sum_{m=M/2}^{M-3}  \frac{d}{dt} \langle \Phi_{N,M}(t),\1^{\le m}\Phi(t)\rangle \Big| \le C_{t,\eps} N^\eps \Big( \sqrt{\frac{M}{N}} + \frac{1}{M} \Big) , \quad \forall \eps>0.
\end{align*}
Consequently,
\begin{align*}
&\Re \frac{1}{M/2-2}\sum_{m=M/2}^{M-3}  \langle \Phi_{N,M}(t),\1^{\le m}\Phi(t)\rangle\\
&\ge \Re \frac{1}{M/2-2}\sum_{m=M/2}^{M-3}  \langle \Phi_{N,M}(0),\1^{\le m}\Phi(0)\rangle - C_{t,\eps} N^\eps \Big( \sqrt{\frac{M}{N}} + \frac{1}{M} \Big).
\end{align*}
Moreover, 
$$ \langle \Phi_{N,M}(0),\1^{\le m}\Phi(0)\rangle \ge \langle \Phi(0),\1^{\le m}\Phi(0)\rangle \ge 1- \frac{C}{M}.$$ 
Thus
\begin{align*}
&\Re \frac{1}{M/2-2}\sum_{m=M/2}^{M-3}  \langle \Phi_{N,M}(t),\1^{\le m}\Phi(t)\rangle\ge 1- \frac{C}{M} - C_{t,\eps} N^\eps \Big( \sqrt{\frac{M}{N}} + \frac{1}{M} \Big).
\end{align*}
Averaging \eqref{eq:many-pa-PhiNM-Phi} over $m\in [M/2,M-3]$ and combining with the latter estimate, we arrive at
$$
\Re \langle \Phi_{N,M}(t),\Phi(t)\rangle \ge 1- C_{t,\eps} N^\eps \Big( \sqrt{\frac{M}{N}} + \frac{1}{M} \Big), \quad \forall \eps>0,
$$
and hence 
$$
\| \Phi_{N,M}(t) - \Phi(t)\|^2 \le 2(1- \Re \langle \Phi_{N,M}(t),\Phi(t)\rangle )\le C_{t,\eps} N^\eps \Big( \sqrt{\frac{M}{N}} + \frac{1}{M} \Big), \quad \forall \eps>0.
$$
\end{proof}
Now we are ready to conclude
\begin{proof}[Proof of Theorem \ref{thm:main} in case $d=2$]  Since $U_N(t)$ is a unitary operator from $\gH^N$ to $\cF_+^{\le N}$, we have the identity
$$
\| \Psi_N(t) - U_N(t)^* \1^{\le N}\Phi(t)\| = \| U_N(t) \Psi_N(t) - \1^{\le N} \Phi(t)\| \le  \| \Phi_N(t) -  \Phi(t)\|.
$$
Using the triangle inequality 
$$
 \| \Phi_N(t) -  \Phi(t)\| \le  \| \Phi_N(t) -  \Phi_{N,M}(t)\|+ \| \Phi_{N,M}(t) -  \Phi(t)\|,
$$
where $\Phi_{N,M}$ is introduced in \eqref{eq:def-PhiNM} with $M=N^{1-\delta}$, and using the comparison results in Lemmas \ref{lem:PhiN-PhiNM}, \ref{lem:com-PhiNM-Phi}, we find that
$$
\| \Phi_N(t) -  \Phi(t)\|^2 \le C_{t,\eps} N^\eps \Big( \frac{ N^\beta}{M}  + \frac{1}{M^{1/2}} + \sqrt{\frac{M}{N}}    \Big), \quad \forall \eps>0.
$$
By choosing
$$ M=N^{\frac{2\beta+1}{3}}$$
we obtain
$$
\| \Phi_N(t) -  \Phi(t)\|^2 \le C_{t,\eps} N^\eps N^{\frac{\beta-1}{3}}, \quad \forall \eps>0.
$$
This ends the proof.
\end{proof}

\begin{proof}[Proof of Corollary \ref{cor}] It is well-known that the trace norm controls the partial trace norm, and hence from the norm convergence we can deduce easily that 
$$
\lim_{N\to \infty} \Tr \Big| \gamma_{\Psi_N(t)}^{(1)} - |u(t)\rangle \langle u(t)| \Big| =0
$$
(see \cite[Corollary 2]{LewNamSch-15} for a detailed explanation). Moreover, since $w_N \wto a \delta_0$ weakly with $a=\int w$, it is straightforward to check that the Hartree solution converges strongly in $L^2(\R^d)$ to the nonlinear Sch\"odinger evolution in \eqref{eq:NLS} with the initial value $\varphi(0,x)=u(0,x)$ and the phase factor
$$
\mu(t)= \lim_{N\to \infty}\mu_N(t)=\frac{a}{2} \int_{\R^d} |u(t,x)|^4 \d x. 
$$
Thus we conclude that
$$
\lim_{N\to \infty} \Tr \Big| \gamma_{\Psi_N(t)}^{(1)} - |\varphi(t)\rangle \langle \varphi(t)| \Big| =0
$$
To go from equation \eqref{eq:NLS} to equation \eqref{eq:NLS-0phase}, we only need to use a gauge transformation 
$$ \varphi(t) \mapsto e^{-i\int_0^t \mu (s)ds} \varphi(t).$$
This transformation, however, does not change the projection $|\varphi(t)\rangle \langle \varphi(t)|$ and the desired conclusion follows. 
\end{proof}

\section{Proof of Theorem \ref{thm:main} ($d=1$)} \label{sec:last-d1}

In this section, we prove Theorem \ref{thm:main} in case $d=1$. The scheme of the proof remains the same as in the case of $d=2$, but the proof for $d=1$ is somewhat simpler, thanks to the stronger Sobolev's inequality in one-dimension. Below we will sketch the main steps again and point out where the differences lie.

First, we consider  the Hartree equation \eqref{eq:Hartree-equation}. When  $w_N(x)=N^{\beta} w(N^\beta x)$ with $w\in L^1(\R)$ and $\beta>0$, we can show that for every $u(0)\in H^3(\R)$ the equation \eqref{eq:Hartree-equation} has a unique solution $u(t,.)$ in $H^3(\R)$ and for all $t>0$ we have the bounds
$$
\| u(t,\cdot)\|_\infty\le \|u(t,\cdot)\|_{H^1(\R)} \le C, \quad  \|u(t,\cdot)\|_{H^2(\R)} \le Ce^{Ct}, \quad \|\partial_t u(t,\cdot)\|_\infty \le Ce^{Ct}
$$
where the constant $C>0$ is independent of $t$ and $N$ (it depends only on $\|w\|_{L^1}$ and $\|u(0)\|_{H^3(\R)}$). This result is obtained by following the same strategy as in the proof of Lemma \ref{lem:Hartree}, plus Sobolev's inequality in 1D (see e.g. \cite{LieLos-01})
$$
\|f\|_{L^\infty(\R)}^2 \le \|f\|_{L^2(\R)}  \|\partial_x f\|_{L^2(\R)}, \quad \forall f\in H^1(\R).
$$

Next, we consider Lemma \ref{lem:Sobolev-inverse-K2}. By following the proof of Lemma \ref{lem:Sobolev-inverse-K2} and using 
\begin{equation*}
\int_{\R^d}  (1+|2\pi p|^2)^{-1}  \d p  <\infty
\end{equation*}
when $d=1$, we obtain the improved version of Lemma \ref{lem:Sobolev-inverse-K2} when $d=1$: 
\begin{align*}
\|(1-\Delta_x)^{-1/2} K_2(t, \cdot, \cdot)\|^2_{L^2} + \|(1-\Delta_x)^{-1/2} \partial_t K_2(t, \cdot, \cdot)\|^2_{L^2} \le C_{t}.
\end{align*}
This leads to an improved estimate in Lemma \ref{lem:bHt-dbHt}, namely for $d=1$ and any $\eta>0$ we have the bounds 
\begin{equation}
\begin{aligned}  \label{eq:1DbHbounds}
\pm \Big( \bH(t) + \dGamma(\Delta) \Big) &\le \eta \dGamma(1-\Delta) + C_{t}(1+\eta^{-1})\cN,\\
\pm \partial_t \bH(t) &\le \eta \dGamma(1-\Delta) +   C_{t}(1+\eta^{-1})\cN,\\
\pm i[\bH(t),\cN] &\le \eta \dGamma(1-\Delta) +   C_{t}(1+\eta^{-1})\cN. 
\end{aligned}
\end{equation}
Consequently, we have immediately an improved kinetic estimate for the Bogoliubov dynamics $\Phi(t)$:
\bq \label{eq:1DPhikinetic}
\big \langle \Phi(t), \dGamma(1-\Delta)  \Phi(t) \big\rangle \le  C_{t}. 
\eq

Now we consider the operator inequality \eqref{eq:Sobolev_interaction}. When $d=1$, Sobolev's inequality implies that 
\begin{equation}\label{eq:Sobolev1D}
|w_N(x-y)| \le C (1-\Delta_x). 
\end{equation}
As a consequence we are able to bound the error term $\cE_N$ in Lemma \ref{lem:Bog-app} in terms of the kinetic energy on the whole truncated Fock space (not only on the sector with less than $m <N$ particles as for $d=2$). The one dimensional counterpart of Lemma \ref{lem:Bog-app} reads, for any $\eta>0$, 
\begin{equation}
\begin{aligned} \label{eq:1DcEbounds}
\pm \1^{\le N}\cE_{N}(t)\1^{\le N} &\le \eta \dGamma(1-\Delta) + C_{t}(1+\eta^{-1})\cN \\
\pm \1^{\le N} \partial_t \cE_{N}(t)\1^{\le N}  &\le \eta \dGamma(1-\Delta) + C_{t}(1+\eta^{-1})\cN \\
\pm  \1^{\le N} i [\cE_{N}(t),\cN]\1^{\le N}  &\le\eta \dGamma(1-\Delta) + C_{t}(1+\eta^{-1})\cN.
\end{aligned}
\end{equation}
The bounds \eqref{eq:1DbHbounds} and \eqref{eq:1DcEbounds} imply that we can apply the Gr\"onwall argument to the quantity $\langle \Phi_N(t), (\cG_N(t)+C_t\cN)\Phi_N(t)\rangle$ and derive a better, $\beta$-independent bound on the kinetic energy of the full many-body dynamics $\Phi_N(t)$, i.e.
\bq \label{eq:1DPhiNkinetic}
\langle \Phi_N(t), \dGamma(1-\Delta) \Phi_N(t) \rangle \le C_{t}.
\eq

The bounds \eqref{eq:1DPhikinetic} and \eqref{eq:1DPhiNkinetic}  imply in particular that now, when $d=1$, the bounds in Lemmas \ref{lem:PhiN-PhiNM} and \ref{lem:com-PhiNM-Phi} become independent of $\beta$ and $\eps$, namely
$$\|\Phi_{N}(t)-\Phi_{N,M}(t)\|^2\le \frac{C_{t}}{M^{1/2}}$$
and
$$\|\Phi_{N,M}(t)-\Phi(t)\|^2 \le C_{t}  \Big( \sqrt{\frac{M}{N}} + \frac{1}{M} \Big).
$$ 
Then we can proceed exactly as in the proof for $d=2$ in Section \ref{sec:main-proof} and arrive at the desired final estimate.

\end{document}